\documentclass[10pt,twocolumn,journal]{IEEEtran}

\usepackage{subfig,float}
\usepackage{amsmath, mathrsfs, mathtools}
\usepackage{amsfonts}
\usepackage{amssymb}
\usepackage{graphicx}
\usepackage{color}
\usepackage{multirow}
\usepackage{graphicx}
\usepackage{stmaryrd}
\usepackage{mathtools}
\usepackage{textgreek}

\def\Fxi{\clF_{\text{\textxi,\texttheta}}}
\def\Fnu{\clF_{\text{\textnu}}}
\def\Dxirho{\clD_{\text{\textvarsigma,\textrho}}}

\def\ripa{{\epsilon}}
\def\ripb{{\delta}}
\def\ripc{{\mu}}
\def\kapp{{\kappa}}

\def\vtheta{{\vartheta}}
\def\vsigma{{\varsigma}}
\def\vrho{{\varrho}}
\def\vpi{{\varpi}}
\def\lambdam{{\boldsymbol{\lambda}}}

\allowdisplaybreaks

\usepackage{amsfonts}
\usepackage{amssymb}
\usepackage{graphicx}
\usepackage{color}
\usepackage{multirow}
\usepackage{graphicx}

\newcommand{\subparagraph}{}
\usepackage[compact]{titlesec}
\titlespacing*{\section}{10pt}{1.2\baselineskip}{1.1\baselineskip}

\titlespacing\section{0pt}{10pt plus 4pt minus 2pt}{5pt plus 2pt minus 2pt}
\titlespacing\subsection{0pt}{8pt plus 4pt minus 2pt}{3pt plus 2pt minus 2pt}
\titlespacing\subsubsection{0pt}{6pt plus 4pt minus 2pt}{3pt plus 2pt minus 2pt}

\allowdisplaybreaks

\newcommand{\myhash}{%
  {\settoheight{\dimen0}{C}\kern-.05em\, \resizebox{!}{\dimen0}{\raisebox{\depth}{\#}}}}

\usepackage{caption}

\usepackage[numbers,sort&compress]{natbib}

\usepackage{multicol}


\usepackage{mathbbol}

\def\mindex#1{\index{#1}}

\def\sq{\hbox{\rlap{$\sqcap$}$\sqcup$}}
\def\qed{\ifmmode\sq\else{\unskip\nobreak\hfil
\penalty50\hskip1em\null\nobreak\hfil\sq
\parfillskip=0pt\finalhyphendemerits=0\endgraf}\fi\medskip}

\long\def\defbox#1{\framebox[.9\hsize][c]{\parbox{.85\hsize}{%
\parindent=0pt
\baselineskip=12pt plus .1pt      
\parskip=6pt plus 1.5pt minus 1pt 
 #1}}}

\long\def\beginbox#1\endbox{\subsection*{}%
\hbox{\hspace{.05\hsize}\defbox{\medskip#1\bigskip}}%
\subsection*{}}

\def\endbox{}

\def\diag{{\text{diag}}}

\def\tr{\mathsf{tr}}

\newsavebox{\junk}
\savebox{\junk}[1.6mm]{\hbox{$|\!|\!|$}}

\def\limsup{\mathop{\rm lim\ sup}}
\def\liminf{\mathop{\rm lim\ inf}}
\def\argmin{\mathop{\rm arg\, min}}

\newcommand{\field}[1]{\mathbb{#1}}

\def\ZZ{\field{Z}}

\def\bE{{\mathbb E}}
\def\bF{{\mathbb F}}

\def\bP{{\mathbb P}}

\def\bR{{\mathbb R}}

\def\bbb{{\mathbb b}}

\def\bfA{{\bf A}}
\def\bfB{{\bf B}}

\def\bfH{{\bf H}}
\def\bfI{{\bf I}}

\def\bfL{{\bf L}}

\def\bfS{{\bf S}}
\def\bfT{{\bf T}}
\def\bfU{{\bf U}}

\def\bfX{{\bf X}}
\def\bfY{{\bf Y}}
\def\bfZ{{\bf Z}}

\def\bfb{{\bf b}}

\def\bfe{{\bf e}}

\def\bfg{{\bf g}}

\def\bfw{{\bf w}}
\def\bfx{{\bf x}}
\def\bfy{{\bf y}}
\def\bfz{{\bf z}}

\def\scrE{{\mathscr{E}}}

\def\ttA{{\mathtt{A}}}

\def\ttC{{\mathtt{C}}}

\def\ttG{{\mathtt{G}}}

\def\ttT{{\mathtt{T}}}

\def\sfF{{\sf F}}
\def\sfG{{\sf G}}

\def\sfM{{\sf M}}
\def\sfN{{\sf N}}

\def\sfQ{{\sf Q}}
\def\sfR{{\sf R}}

\def\bfmath#1{{\mathchoice{\mbox{\boldmath$#1$}}%
{\mbox{\boldmath$#1$}}%
{\mbox{\boldmath$\scriptstyle#1$}}%
{\mbox{\boldmath$\scriptscriptstyle#1$}}}}

\def\bfmY{\bfmath{Y}}

\def\bfmhhaY{\bfmath{\hhaY}} 
\def\bfmhhaY{\hbox to 0pt{$\widehat{\bfmY}$\hss}\widehat{\phantom{\raise 1.25pt\hbox{$\bfmY$}}}}

\def\til={{\widetilde =}}

\def\clB{{\cal B}}
\def\clC{{\cal C}}
\def\clD{{\cal D}}
\def\clE{{\cal E}}
\def\clF{{\cal F}}
\def\clG{{\cal G}}
\def\clH{{\cal H}}
\def\clJ{{\cal J}}
\def\clI{{\cal I}}

\def\clL{{\cal L}}

\def\clN{{\cal N}}

\def\clR{{\cal R}}
\def\clS{{\cal S}}

\def\clW{{\cal W}}
\def\clX{{\cal X}}
\def\clY{{\cal Y}}
\def\clZ{{\cal Z}}

 \def\FRAC#1#2#3{\genfrac{}{}{}{#1}{#2}{#3}}

\def\ddtp{{\mathchoice{\FRAC{1}{d^{\hbox to 2pt{\rm\tiny +\hss}}}{dt}}%
{\FRAC{1}{d^{\hbox to 2pt{\rm\tiny +\hss}}}{dt}}%
{\FRAC{3}{d^{\hbox to 2pt{\rm\tiny +\hss}}}{dt}}%
{\FRAC{3}{d^{\hbox to 2pt{\rm\tiny +\hss}}}{dt}}}}

\def\average#1,#2,{{1\over #2} \sum_{#1}^{#2}}

\def\eye(#1){{\bf(#1)}\quad}

\newtheorem{theorem}{{\bf Theorem}}
\newtheorem{proposition}{{\bf Proposition}}

\newtheorem{definition}{{\bf Definition}}
\newtheorem{condition}{{\bf Condition}}
\newtheorem{remark}{{\bf Remark}}
\newtheorem{example}{{\bf Example}}

\newtheorem{lemma}[theorem]{{\bf Lemma}}

\def\eq#1/{(\ref{e:#1})}

\newcommand{\inp}[2]{{\langle #1, #2 \rangle}}

\newcommand{\beqn}[1]{\notes{#1}%
\begin{eqnarray} \elabel{#1}}

\newcommand{\eeqn}{\end{eqnarray} }

\newcommand{\beq}[1]{\notes{#1}%
\begin{equation}\elabel{#1}}

\newcommand{\eeq}{\end{equation}}

\def\bdes{\begin{description}}
\def\edes{\end{description}}

\newcounter{rmnum}

\newcounter{anum}

{\end{list}}

\def\ass(#1:#2){(#1\ref{#1:#2})}

\def\ritem#1{
\item[{\sf \ass(\current_model:#1)}]
}

\newenvironment{recall-ass}[1]{%
\begin{description}
\def\current_model{#1}}{
\end{description}
}

\usepackage{tikz}
\usepackage{pgfplots,tikz-3dplot}
\pgfplotsset{compat=newest}
\usetikzlibrary{patterns}
\usetikzlibrary{positioning}
\usetikzlibrary{datavisualization}
\usetikzlibrary{datavisualization.formats.functions}
\usetikzlibrary{backgrounds}
\usetikzlibrary{shapes,snakes}
\usepgfplotslibrary{fillbetween}

\usepgfplotslibrary{units}
\usetikzlibrary{spy}
\usepackage{pgfplotstable}

\usepgfplotslibrary{external} 
\usepackage{float}
\usepackage{afterpage}
\usepackage{balance}

\def\snr{{\mathsf{snr}}}

\usepackage[none]{hyphenat}

\def\vol{{\text{vol}}}

\def\name{{Unlabeled Ordered Sampling}}
\def\namesh{{UOS}} 

\def\algname{{$\mathtt{AltMin}$}}

\long\def\comment#1{}

\newcommand{\Gammam}{\hbox{\boldmath$\Gamma$}}
\newcommand{\Lambdam}{\hbox{\boldmath$\Lambda$}}

\newcommand{\Sigmam}{\hbox{\boldmath$\Sigma$}}

\newcommand{\transp}{{\sf T}}
\renewcommand{\vec}{{\rm vec}}

\newcommand{\range}[2]{{\text{$#1$\,:\,$#2$}}}

\title{Signal Recovery from Unlabeled Samples}
\author{Saeid Haghighatshoar,  \IEEEmembership{Member, IEEE}, Giuseppe Caire, \IEEEmembership{Fellow, IEEE}   
\thanks{The authors are with the Communications and Information Theory Group, Technische Universit\"{a}t Berlin (\{saeid.haghighatshoar, caire\}@tu-berlin.de).}
\thanks{A short version of this paper was presented in IEEE International Symposium on Information Theory (ISIT), 2017, Aachen, Germany.}
}

\begin{document}

\maketitle

\begin{abstract}
In this paper, we study the recovery of a signal from a set of noisy linear projections (measurements), when such projections are unlabeled, that is, the correspondence between the measurements and the set of projection vectors (i.e., the rows of the measurement matrix) is not known a priori. We consider a special case of unlabeled sensing referred to as \textit{Unlabeled Ordered Sampling} (\namesh) where the ordering of the measurements is preserved. 
We identify a natural duality between this problem and  classical \textit{Compressed Sensing} (CS), where we show that the unknown support (location of nonzero elements) of a sparse signal in CS corresponds to the unknown indices of the measurements  in \namesh. While in CS it is possible to recover a sparse signal from an under-determined set of linear equations (less equations than the signal dimension), successful recovery in \namesh\ requires taking more samples than the dimension of the signal. Motivated by this duality, we develop a \textit{Restricted Isometry Property} (RIP) similar to that in CS. We also design a low-complexity Alternating Minimization algorithm that achieves a stable signal recovery under the established RIP.  We analyze our proposed algorithm for different signal dimensions and number of measurements theoretically and  investigate its performance empirically via  simulations. The results are a reminiscent of  phase-transition similar to that occurring in CS.
\end{abstract}

\begin{keywords}
Unlabeled Sensing, Compressed Sensing, Alternating Minimization algorithm.
\end{keywords}

\section{Introduction}
The recovery of a vector-valued signal $\bfy \in \bR^k$ from a set of linear and possibly noisy measurements $\bfx=\bfB \bfy +\bfw$, with an $n\times k$ measurement matrix $\bfB$, is the classical problem of linear regression in statistical inference and is arguably the most widely-studied problem in statistics, mathematics, and computer science. For $n\geq k$, one has an over-determined set of noisy linear equations, and the \textit{Maximum Likelihood} (ML) estimate of $\bfy$ under the additive Gaussian noise $\bfw$ is given by the well-known method of \textit{least squares}. For $n<k$, in contrast, one deals with an under-determined set of linear equations, which is only solvable if some additional a priori information about the signal $\bfy$ is available. For example, the whole field of Compressed Sensing (CS) deals with the recovery of the signal $\bfy$ when it is sparse or more generally compressible, i.e., it has only a few significantly large coefficients when represented in a suitable basis \cite{donoho2006compressed, candes2006near, candes2005decoding}. 

Almost all the past research in linear regression mainly deals with exploiting the underlying signal structure, whereas it is generally assumed that the regression matrix $\bfB$ is  fully known. In practice, the matrix $\bfB$ is implemented through a measurement device, where due to physical limitations, there might be some uncertainty or mismatch between the intended matrix $\bfB$ and the one realized via measurement devices. This has resulted in a surge of interest in generalized linear regression problems in which the matrix $\bfB$ is mismatched \cite{herman2010general,herman2010mixed} or is known only up to some unknown transformation. In this paper, we are interested in the \textit{unlabeled sensing} case, where the observation model is given by
\begin{align}\label{ul_samp}
{\bfx=\bfS \bfB \bfy +\bfw,}
\end{align}
where $\bfB$ is a completely known $n \times k$ matrix, and where $\bfS$ is an unknown $m \times n$ matrix with 0-1 components that has only a single $1$ in each row and samples (selects) $m$ out of $n$ elements of $\bfB \bfy$. Although $\bfS$ is not known, it belongs to an a priori  known set of selection matrices $\clS$. Identifying $\bfS$ in \eqref{ul_samp}, therefore, corresponds to associating the  measurements $\bfx$ to their corresponding linear projections in $\bfB$. Once $\bfS$ is identified, \eqref{ul_samp} reduces to a  linear regression problem whose solution can be obtained via standard techniques.

In \cite{unnikrishnan2015unlabeledarxiv, unnikrishnan2015unlabeled}, a variant of  this problem was studied when $\clS$ is the set of all permutations of $n$ measurements taken by $\bfB$. It was shown that if the measurement matrix $\bfB$ is generated randomly, any arbitrary $k$-dim signal can be recovered from a set of $n=2k$ noiseless unlabeled measurements, where this bound was shown to be tight. In \cite{pananjady2016linear}, the authors studied a similar problem but rather than recovering the unknown signal $\bfy$, they obtained a scaling law of the \textit{Signal-to-Noise Ratio} (SNR) required for identifying the unknown permutation matrix in $\clS$, where they showed that the required SNR increases logarithmically with the signal dimension. A \textit{Multiple Measurement} version of the problem defined by $\bfy^i=\bfS \bfB \bfx^i +\bfw^i$, $i=1, \dots, l$, where the signals $\bfx^i$ might vary but $\bfS$ remains the same across all $l$ measurements, was recently studied in \cite{pananjady2017denoising, hsu2017linear} and was proved to yield a signal recovery at a finite SNR provided that $l$ is sufficiently large.

A practical scenario well-modeled by \eqref{ul_samp} is sampling in the presence of jitter \cite{balakrishnan1962problem}, in which $\clS$ consists of 0-1 matrices with 1s as their diagonal elements and with some off-diagonal 1s representing the location of the jittered samples. A similar problem arises in molecular channels \cite{rose2012timing} and in the reconstruction of phase-space dynamics of linear/nonlinear dynamical systems \cite{fung2016dynamics} because of synchronization issues. 
Unlabeled regression \eqref{ul_samp} is also encountered in noncooperative multi-target tracking, e.g., in radar, where the receiver only observes the unlabeled data associated to all the targets, thus, $\clS$ consists of the set of all permutations corresponding to all possible data-target associations \cite{skolnik1980}.
A quite similar scenario arises in robotics and a well-known classical problem is \textit{Simultaneous Localization and Mapping} (SLAM) where  robots measure their relative coordinates and recovery of the underlying geometry requires suitably permuting the  data. 

A different line of work well-modeled by \eqref{ul_samp} is the genome assembly problem from shotgun reads \cite{huang1999cap3, bresler2013optimal} in which a sequence $y\in \{\ttA, \ttT , \ttG, \ttC \}^k$ of length $k$ is assembled from an unknown permutation of its sub-sequences called ``reads''.  {Designing efficient recovery (assembly) algorithms is still an active research area (see \cite{bresler2013optimal} and the refs. therein.).}
Communication over the  classical noisy deletion channel is another example of  \eqref{ul_samp}, where $\bfB$ represents the linear encoding matrix with elements belonging to a finite field $\bF_q$, for some prime number $q$, where $\bfy\in \bF_q^k$ denotes the $k$-dim vector containing $k$ information symbols, where $\bfw\in \bF_q^m$ is the additive noise of the channel, and where $\clS$ consists of all selection operators that keep only $m\leq n$ out of $n$ encoded symbols in $\bfB\bfy$ while preserving their order. In particular, in contrast with the erasure channel, where the location of erased symbols is known, in a deletion channel the location of deleted symbols is missing. 
 {This makes  designing good encoding and decoding techniques as well as identifying the capacity quite challenging \cite{mitzenmacher2009survey}. }

\begin{figure}[t]
\centering
\includegraphics{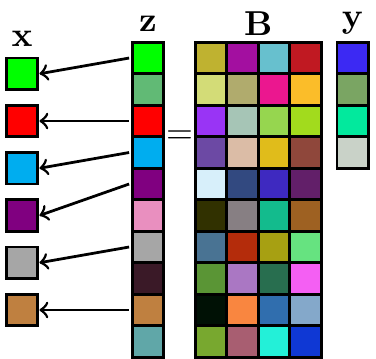}
\caption{\name.}
\label{fig:UOS}
\end{figure}

\subsection{Contribution}
Since satisfactory efficient algorithms for solving the unlabeled sampling problem in \eqref{ul_samp} are yet unknown,  we make 
progress towards this goal by addressing a relevant subproblem that we refer to as \textit{\name} (\namesh), where $\clS$ is the set of all 0-1 \textit{ordered sampling} matrices that select only $m$ out of $n$ components of $\bfB\bfy$ for some $m\leq n$ while preserving their relative order. This occurs in many practical scenarios (e.g., in a deletion channel) and is illustrated in Fig.\,\ref{fig:UOS}.
We discover a duality between this problem and the CS problem \cite{donoho2006compressed, candes2006near, candes2005decoding}, where the unknown location of samples in the former corresponds in a natural way to the unknown location of nonzero coefficients of the signal in the latter. 
To the best of our knowledge, this is the first paper addressing the underlying duality connection between the unlabeled sensing in \eqref{ul_samp} and  classical CS.  Designing a low-complexity algorithm for recovering the desired signal from its unlabeled samples is generally considered to be a  challenging problem \cite{unnikrishnan2015unlabeledarxiv, unnikrishnan2015unlabeled,pananjady2016linear, balakrishnan1962problem}. In particular, in most practically relevant situations, $\clS$ is a very large set such that a naive exhaustive search over $\clS$ would be unfeasible.  
In this paper, however, we are able to exploit the underlying ordering in the \namesh\ case  to design an efficient \textit{Alternating Minimization} recovery algorithm (\algname). We  analyze the performance of our proposed algorithm theoretically and investigate its performance empirically via  simulations.

\subsection{Notation}
We denote vectors by  boldface small letters (e.g., $\bfx$), matrices by boldface capital letters (e.g., $\bfX$), scalars by non-boldface letters (e.g., $x$), and sets by capital calligraphic letters (e.g., $\clX$).  For integers $k_1, k_2 \in \ZZ$, we use the shorthand notation $[\range{k_1}{k_2}]=\big \{k_1, k_1+1, \dots, k_2\}$, where the set is empty when $k_2 < k_1$, and use the simplified notation $[k_1]$ for $[\range{1}{k_1}]$. We denote the $k$-th component of a vector $\bfx$ by $\bfx_k$ and a subvector of $\bfx$ with indices in the range $\range{k_1}{k_2}$ by $\bfx_{\range{k_1}{k_2}}$. We denote the element of a matrix $\bfX$ at row $r$ and column $c$ by $\bfX_{r,c}$ and use an indexing notation similar to that for vectors for submatrices of $\bfX$, namely,  $\bfX_{r,c}$, $\bfX_{r,\range{c_1}{c_2}}$, $\bfX_{\range{r_1}{r_2},c}$, and $\bfX_{\range{r_1}{r_2}, \range{c_1}{c_2}}$. 
We denote the Kronecker product of a $p\times q$ matrix $\bfX$ and an $r\times s$ matrix $\bfY$ by a $pr\times qs$ matrix $\bfX\otimes \bfY$ that has $p\times q$ blocks with the $ij$-th block, $i\in[p], j\in[q]$, given by the $r\times s$ matrix $\bfX_{i,j} \bfY$. 
We use $\tr(.)$ for the trace  and $(.)^\transp$ for the transpose operator. We represent the inner product between two  vectors $\bfx$ and $\bfy$ and two matrices $\bfX$ and $\bfY$, with $\inp{\bfx}{\bfy}=\bfx^\transp \bfy$ and $\inp{\bfX}{\bfY}=\tr(\bfX^\transp \bfY)$  respectively, and use $\|\bfx\|=\sqrt{\inp{\bfx}{\bfx}}$ for the $l_2$ norm of a vector $\bfx$, and $\|\bfX\|_\sfF=\sqrt{\inp{\bfX}{\bfX}}$ for the Frobenius norm of a matrix $\bfX$. 
We denote a diagonal matrix with the elements $(x_1, \dots, x_k)$ by $\diag(x_1, \dots, x_k)$ and the identity matrix of the order $k$ with $\bfI_k$.
We represent a sequence of vectors and sequence of matrices by upper indices, e.g., $\bfx^1, \bfx^2, \cdots$ and $\bfX^1, \bfX^2,\cdots$. We denote a Gaussian distribution with a mean $\mu$ and a variance $\sigma^2$ by $\sfN(\mu, \sigma^2)$. We use $O(.)$ and $o(.)$ for the big-O and the small-O respectively.

\section{Statement of the Problem}
In this section, we  start from the more familiar CS problem. Then, we introduce the \namesh\ and make a duality connection between the two.

\subsection{Basic Setup}
Let $n$ and $m$ be positive integers with $m\leq n$, and let $\clI=\{i_1, \dots, i_m\} \subseteq [n]$ be a subset of $[n]$ consisting of ordered elements of $[n]$ satisfying $i_l<i_{l+1}$. We define the \textit{lift-up} operator associated to $\clI$ as a linear map from $\bR^m$ to $\bR^n$ given by the $\{0,1\}$-valued $n\times m$ tall matrix $\bfL$ with
\begin{align}
\bfL_{p,q}=\left \{ \begin{array}{ll} 1 & (p,q)=(i_l,l) \text{ for } l\in[m],\\ 0 & \text{otherwise.} \end{array} \right.
\end{align}
The operator $\bfL$ embeds $m$ components of $\bfx=(x_1, \dots, x_m)^\transp$ into the index set $\clI$ in the $n$-dim vector $\bfL\bfx$, while keeping their relative order, i.e., $(\bfL\bfx)_{i_l}=x_l$ for $l \in [m]$, and fills the rest with $0$. For example, for $n=4$, $m=3$, $\bfx=(x_1,x_2,x_3)$,  and $\clI=\{1,3,4\}$, we have $\bfL\bfx=(x_1,0,x_2,x_3)$.
We define the collection of all ${n \choose m}$ lift-up operators by $\clL_{n,m}$.

\subsection{Compressed Sensing}
In CS \cite{donoho2006compressed, candes2006near, candes2005decoding}, the goal is to recover a sparse or compressible signal by taking  less measurements than its dimension. We call a signal $\bfz \in \bR^n$ $m$-sparse ($m$-compressible) if it has only $m$ nonzero (significantly large) components. For simplicity, we focus on sparse rather than compressible signals.
We fix $m,n$ with $m\leq n$ as before. We define an instance of the CS problem for an $m$-sparse signal $\bfz \in \bR^n$ by the triple $(\bfx, \bfL, \bfA)$, where $\bfx\in \bR^m$ denotes the nonzero elements of $\bfz$, where $\bfL\in \clL_{n,m}$ encodes the location of these nonzero elements, and where $\bfA$ is the $k\times n$ matrix whose rows correspond to $k$ linear measurements. 
The  $m$-sparse signal $\bfz$ is generated by embedding the $m$ components of $\bfx$  via the lift-up matrix $\bfL$ as $\bfz=\bfL\bfx$, where it is seen that $\bfz$, albeit being $n$-dim, has at most $m$ nonzero components. This is illustrated in Fig.\,\ref{fig:dual_fig}. 
In CS, the sparse signal $\bfz$ is sampled via the measurement matrix $\bfA$ producing $k$  measurements $\bfy=\bfA \bfz$.  The goal is  to recover the unknown $(\bfx,\bfL)$ or equivalently $\bfz=\bfL\bfx$, from the known $(\bfy, \bfA)$.
The crucial idea in CS is that by exploiting the underlying sparsity, $\bfz$ can be recovered by taking less samples than its embedding dimension $n$.  
The practically interesting  regime of parameters in  CS is given by $m\leq k \leq n$, where the number of measurements $k$ is more than the number of nonzero coefficients of the signal $m$ but much less than the embedding dimension $n$.
\begin{figure}[t]
\centering
\includegraphics{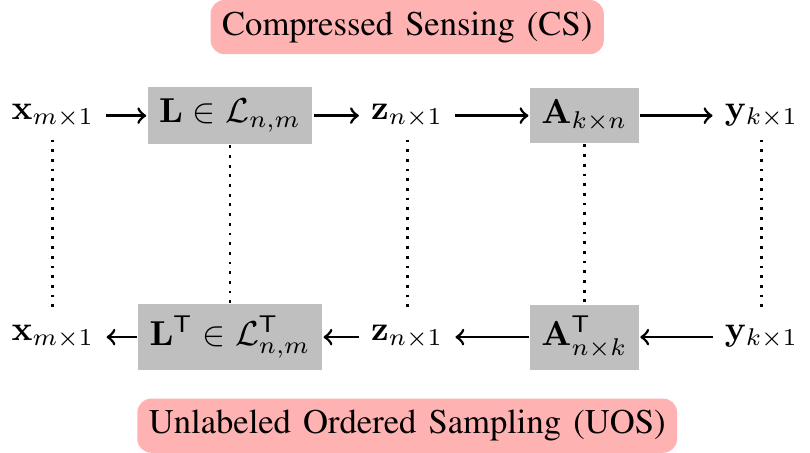}
\caption{Comparison between Compressed Sensing and Unlabeled Ordered Sensing.}
\label{fig:dual_fig}
\end{figure}
\subsection{\name}
By changing the role of $\bfy$ and $\bfx$ in the CS problem, we obtain an instance of \namesh\ in \eqref{ul_samp} as the dual problem with $\bfB=\bfA^\transp$ and $\bfS=\bfL^\transp$. This is also illustrated in  Fig.~\ref{fig:dual_fig}. In \namesh, a $k$-dim signal $\bfy$ is oversampled via the tall $n \times k$ matrix $\bfB=\bfA^\transp$, which gives $n$ measurements  $\bfz=\bfA^\transp \bfy$. The resulting over-complete measurements ($n\geq k$) in $\bfz$ are subsampled by the $m\times n$ matrix $\bfS=\bfL^\transp$, which selects only $m$ out of $n$ measurements in $\bfz$ and yields  the  unlabeled samples  $\bfx$. The goal is to recover the unknown signal $\bfy$ from the known  $(\bfx,\bfB)$, without knowing $\bfS$ explicitly. Compared with  CS, where the support (location of the nonzero values) of $\bfz$ is unknown, in \namesh\  the indices/labels of the measurements are missing. However, it is not difficult to check that, due to the special structure of $\bfS$, the relative order of the measurements is still preserved.
We define the set of all such selection matrices by $\clS_{m,n}:=\{\bfL^\transp: \bfL \in \clL_{n,m}\}$. In contrast with the lift-up matrices in $\clL_{n,m}$, which embed signals with a lower dimension $m$ in a higher dimension $n$, the selection matrices in $\clS_{m,n}$ reduce the dimensionality by sampling only $m$ out of $n$ components of their input (while keeping the order). 
The practically relevant regime of parameters for  \namesh\ is given by $k\leq m \leq n$, where one needs to take more measurements $m$ than the signal dimension $k$ to overcome the uncertainty caused by the unlabeled sensing.
Motivated by the duality between CS and \namesh, in the next section, we develop a \textit{Restricted Isometry Property} (RIP), which resembles its counterpart in CS \cite{candes2005decoding} and guarantees a stable signal recovery in \namesh. We also design an efficient low-complexity algorithm able to recover the target signal under the  RIP.

\section{Restricted Isometry Property}

\subsection{Basic Setup}
For the rest of the paper, we focus on  \namesh\   illustrated in Fig.\,\ref{fig:UOS}. We consider a $k$-dim signal $\bfy$ and an $n$-dim vector of  measurements $\bfz=\bfB \bfy$ taken via an $n\times k$  matrix $\bfB$, where $\bfB=\bfA^\transp$ with $\bfA$ being the $k\times n$ matrix in the CS variant (see Fig.~\ref{fig:dual_fig}). 
An instance of \namesh\   is given by the triple $(\bfy, \bfS, \bfB)$, where the goal is to recover the unknown signal $\bfy$ from the noisy unlabeled measurements $\bfx=\bfS \bfB \bfy+\bfw$ without having any explicit knowledge about $\bfS$ except that $\bfS\in \clS_{m,n}$. This corresponds to a variant of unlabeled sensing problem  in \eqref{ul_samp} with the set of possible  transformations $\clS$ given by $\clS_{m,n}$. In the rest of the paper, for simplicity, we drop the index $m,n$ and denote $\clS_{m,n}$ by $\clS$.
Using the $\vec$ notation, we can write 
\begin{align}\label{vec_uos}
\bfx=\vec(\bfS \bfB \bfy)+\bfw=(\bfy^\transp \otimes \bfS) \bbb +\bfw,
\end{align}
where $\bbb=\vec(\bfB)$ denotes the vector obtained by stacking the columns of $\bfB$   and where we used the well-known identity ${\vec(\bfX \bfY \bfZ)=(\bfZ^\transp \otimes \bfX) \vec(\bfY)}$ for matrices $\bfX, \bfY, \bfZ$ of appropriate dimensions. We will use the notations $\bfS \bfB \bfy$ and $(\bfy^\transp \otimes \bfS) \bbb$ interchangeably across the paper. Note that $\bbb$ induces a linear map $(\bfy^\transp \otimes \bfS)  \mapsto (\bfy^\transp \otimes \bfS) \bbb$ from the signal set ${\clH:=\{\bfy^\transp \otimes \bfS: \bfy \in \bR^k, \bfS\in \clS\}}$ into $\bR^m$.  We  define the \textit{Signal-to-Noise Ratio} (SNR) in \eqref{vec_uos} by $\snr=\frac{\|(\bfy^\transp \otimes \bfS) \bbb\|^2}{\|\bfw\|^2}=\frac{\|\bfS \bfB \bfy\|^2}{\|\bfw\|^2}$.

\subsection{Restricted Isometry Property on $\clH$}\label{sec:rip_H}
Our goal is to recover the desired signal $\bfy$. Since the signal set $\clH$ is unbounded, a requirement is that at least $\|\bfy\|$ should be feasibly recovered from the measurements $(\bfy^\transp\otimes \bfS) \bbb$. A sufficient condition for this  is the \textit{Restricted Isometry Property} (RIP) over $\clH$, which resembles a similar property  in CS \cite{candes2005decoding}. 
\begin{definition}[RIP over $\clH$]\label{RIP_H}
Let $\bfB$ be a matrix and let $\bbb=\vec(\bfB)$. The linear map $\bfH\mapsto \bfH\bbb$ induced by $\bbb$ satisfies the RIP over $\clH$  with a constant $\ripa \in (0,1)$ if 
\begin{align}\label{rip_H}
(1-\ripa) \|\bfH\|_\sfF^2 \leq \|\bfH \bbb\|^2 \leq (1+\ripa) \|\bfH\|_\sfF^2
\end{align}
holds for all $\bfH \in \clH$. \hfill $\lozenge$
\end{definition}
We prove that under suitable conditions on $m,n,k$, we can obtain a matrix $\bfB$ satisfying the RIP in \eqref{rip_H} by  sampling components of $\bfB$ i.i.d. from the Gaussian distribution. 
This is summarized in the following proposition.
\begin{proposition}\label{rip_1}
Let $\bfB$ be a random matrix with i.i.d. $\sfN(0,1)$ components and let $\bbb=\vec(\bfB)$. Then, there is a constant $c>0$ such that $\bfB$ satisfies  RIP over $\clH$ with a probability larger than $1- 2 (1+\frac{2}{\ripa})^k {n \choose m}e^{-c m \ripa^2}$.  \hfill $\square$
\end{proposition}
\begin{proof}
Let us define the following probability event (for the random realization of $\bbb$):
\begin{align}\label{E_event_for_H}
\clE=\bigcup_{\bfH\in \clH} \Big \{\bbb:  \big | \|\bfH \bbb\|^2 - \|\bfH\|_\sfF^2\big | > \ripa \|\bfH\|_\sfF^2 \Big \}.
\end{align}
We need to prove that $\bP[\clE] \leq 2 (1+\frac{2}{\ripa})^k {n \choose m}e^{-c m \ripa^2}$. Since $\clE$ consists of the union of a continuum of events labeled with $\bfH \in \clH$, the conventional union bound is not immediately applicable. So, we  need to first quantize the labeling set $\clH$ appropriately and then apply the union bound. We do this by using the net argument (see \cite{talagrand2014upper, vybiral2016random} for further discussion on using the net technique). 
We start by first deriving a concentration bound for a fixed $\bfH=\bfy^\transp \otimes \bfS \in \clH$. From $\bfH\bbb=\bfS \bfB \bfy$ and $\bE[\|\bfH \bbb\|^2]=\|\bfH\|_{\sfF}^2$, the subevent of \eqref{E_event_for_H} corresponding to the given $\bfH$ can be written as 
\begin{align}\label{eq_otherform}
\Big \{\bfB: (1-\ripa) m \|\bfy\|^2  \leq \|\bfS \bfB \bfy\|^2 \leq (1+\ripa) m \|\bfy\|^2 \Big \},
\end{align}
where we used the fact {$\|\bfH\|_\sfF^2=m\|\bfy\|^2$}. 
Note that for a fixed $\bfS$ and $\bfy$, the vector $\bfS \bfB \bfy$ is an $m$-dim vector with i.i.d. $\sfN(0,1)$ components. From the Gaussian concentration result \cite{dubhashi2009concentration}, there is a $c>0$ such that 
\begin{align}\label{sy}
\bP\Big [  \Big \{ \bfB: \big |\|\bfS \bfB \bfy\|^2 - m \|\bfy\|^2\big |> \frac{\ripa}{2} m \|\bfy\|^2  \Big \}\Big ] \leq 2 e^{-c m\ripa^2}.
\end{align}
We first generalize this concentration result to all the vectors $\bfy$ inside the unit $l_2$ ball ${\clB_{2}^k:=\{\bfy \in \bR^k: \|\bfy\|\leq1\}}$. We  consider a discrete $\frac{\ripa}{2}$-net (grid) of minimal size $N$ over $\clB_{2}^k$ denoted by $\clN_{\ripa}=\{\bfg^1, \dots, \bfg^N\}$ that satisfies  $\max_{\bfy \in \clB_{2}^k} \min_{\bfg \in \clN_{\ripa}} \|\bfg- \bfy\| \leq \frac{\ripa}{2}$ (see \cite{vybiral2016random} for further details). 
Consider the set of $N$ spheres with centers belonging to the net $\clN_{\ripa}$  each having a radius $\frac{\ripa}{2}$. All these spheres lie inside a sphere of radius $1+\frac{\ripa}{2}$ centered at the origin.  Thus, using the volume inequality $N (\frac{\ripa}{2})^k \vol(\clB_2^k) \leq (1+\frac{\ripa}{2})^k \vol(\clB_2^k)$ for the  $k$-dim unit $l_2$ ball $\clB_{2}^k$, we obtain that such a minimal net consists of at most $N \leq (1+\frac{2}{\ripa})^k$ points. We also have  (see  Lemma 2.19. in \cite{vybiral2016random})
\begin{align}\label{net_app}
\max_{\bfg \in \clN_\ripa}\|\bfS \bfB \bfg\| \leq \max_{\bfy \in \clB_{2}^k} \|\bfS \bfB \bfy\| \leq (1+\frac{\ripa}{2}) \max_{\bfg \in \clN_\ripa}\|\bfS \bfB \bfg\|,
\end{align}
which implies that the operator norm of $\bfS \bfB$ can be well approximated via the points in $\clN_{\ripa}$. From \eqref{net_app}, we obtain
\begin{align}\label{s_bound}
\bP\Big [& \bigcup_{\bfy \in\clB_{2}^k} \Big \{ \bfB: \big |\|\bfS \bfB \bfy\|^2 - m \|\bfy\|^2\big |> m \|\bfy\|^2 \ripa  \Big \}  \Big ] \nonumber\\
&\leq \bP\Big [ \bigcup_{\bfg \in \clN_{\ripa}} \Big \{ \bfB: \big|\|\bfS \bfB \bfg\|^2 - m \|\bfg\|^2\big |> m \|\bfg\|^2\frac{\ripa}{2}  \Big \} \Big ]\nonumber\\
& \stackrel{(i)}{\leq} |\clN_{\ripa}| e^{-c m\ripa^2}=2 (1+\frac{2}{\ripa})^k e^{-c m\ripa^2}.
\end{align}
where in $(i)$ we applied the union bound over $\clN_{\ripa}$ and used the bound \eqref{sy}, which holds for any $\bfy$ and in particular for any $\bfg \in \clN_\ripa$.
Finally applying the union bound over all ${n \choose m}$ possible selection matrices $\bfS\in \clS$, we have 
\begin{align*}
\bP\Big [ &\bigcup_{\bfy\in \clB_2^k ,\bf\bfS \in \clS}  \Big \{ \bfB: \big|\|\bfS \bfB \bfy\|^2 - m\|\bfy\|^2 \big|> m \|\bfy\|^2 \ripa   \Big \}\Big ]\\
& \leq 2|\clS| (1+\frac{2}{\ripa})^k {n \choose m}e^{-c m \ripa^2} =2 (1+\frac{2}{\ripa})^k {n \choose m}e^{-c m \ripa^2},
\end{align*}
which from \eqref{E_event_for_H} implies that $\bP[\clE] \leq 2 (1+\frac{2}{\ripa})^k {n \choose m}e^{-c m \ripa^2}$. This completes the proof.  
\end{proof}

\subsection{Restricted Isometry Property on $\clH-\clH$}\label{sec:rip_H-H}
In terms of signal recovery, we will need a stronger RIP over the Minkowski difference of the signal set $\clH$ defined by $\clH-\clH:=\{\bfH-\bfH': \bfH, \bfH' \in \clH\}$. In this section, our goal is to develop a suitable notion of RIP over $\clH-\clH$. Our approach in this section follows from similar techniques in \cite{vershynin2015estimation}. 

Let $\bfB$ be the $m\times n$ measurement matrix and let $\bbb=\vec(\bfB)$ as before. Similar to the RIP over $\clH$, for the RIP over $\clH-\clH$, it seems reasonable to impose the condition
\small
\begin{align}\label{RIP_diff}
{(1-\ripb)\|\bfH-\bfH'\|_\sfF^2 \leq \|(\bfH-\bfH')\bbb\|^2 \leq (1+\ripb)\|\bfH-\bfH'\|_\sfF^2},
\end{align}
\normalsize
with some $\ripb \in (0,1)$, to hold  for all $\bfH, \bfH' \in \clH$. As we will  see in  Section \ref{sec:sig_rec}, an RIP condition as in \eqref{RIP_diff} is sufficient to guarantee a stable signal recovery for \namesh. In this section, as in Section \ref{sec:rip_H}, we attempt to construct a matrix $\bfB$ satisfying \eqref{RIP_diff} via sampling the components of $\bfB$ randomly. 
Hence, we need to prove that,  in a suitable regime of parameters $m,n,k$, and $\ripb$, any realization of the matrix  $\bfB$ and as a result $\bbb$ satisfies \eqref{RIP_diff} for all $\bfH,\bfH'\in \clH$ with a very high probability. Unfortunately, such a universal concentration result over $\clH - \clH$ does not immediately hold as can be seen from the following simple example\footnote{This should be contrasted with  classical CS, where $\clH$ is the class of $k$-sparse $n$-dim signals, thus, $\clH-\clH$ is a subset of $2k$-sparse signals, and deriving the RIP over $\clH$ and $\clH-\clH$ requires quite similar techniques. In \namesh, in contrast, although it is easy to derive an RIP over $\clH$, obtaining a suitable notion of  RIP over $\clH-\clH$ is more challenging.}.
\begin{example}\label{example_rip}
Let $\bfB$ be a random matrix with i.i.d. $\sfN(0,1)$ components. 
Consider the signals $\bfH=\bfy^\transp \otimes \bfS$ and $\bfH'={\bfy'}^\transp \otimes \bfS'$, where $\bfy=\bfy'$ and where all the rows of $\bfS$ and $\bfS'$ are similar except the last row. For such a case, we have
\begin{align}
(\bfH-\bfH')\bbb=\bfS \bfB \bfy-\bfS'\bfB \bfy'=(0,0, \dots, 0, g)^\transp,\label{non_conc_1}
\end{align}
where $g=(\bfB_{r_m,:}  - \bfB_{r'_m,:}) \bfy$, where $r_m, r'_m \in [n]$ denote the indices of the last rows selected by $\bfS$ and $\bfS'$ respectively. From \eqref{non_conc_1}, it results that for the given $\bfH,\bfH'$, $\|(\bfH-\bfH')\bbb\|^2=|g|^2$, where ${g\sim \sfN(0, 2\|\bfy\|^2)}$ from the independence of the rows of $\bfB$. 
However, $|g|^2$ does not concentrate very well around its mean $\bE[|g|^2]$. This implies that the universal concentration result over $\clH-\clH$ in \eqref{RIP_diff} can not hold for any $\ripb \in (0,1)$. \hfill $\lozenge$
\end{example}

From Example \ref{example_rip}, it is seen that we can not hope to construct, with a high probability,  a matrix $\bfB$ satisfying RIP in \eqref{RIP_diff} under the random sampling of components of $\bfB$. A direct inspection in the Example \ref{example_rip}, however, reveals that the main obstacle to establishing the concentration \eqref{RIP_diff} are those signals $\bfH, \bfH'$ where $\bfy \approx \bfy'$ and $\bfS\approx \bfS'$, thus, $\bfH\approx \bfH'$. 
Since we use the RIP to guarantee a universal stable recovery for all the signals in $\clH$, intuitively speaking, the troublesome cases $\bfH \approx \bfH'$ for \eqref{RIP_diff} are not problematic at all in terms of signal recovery. 
To take this into account, we develop a modified version of  RIP over $\clH -\clH$ in \eqref{RIP_diff} up to a fixed  precision. 

We need some notation first.
We  define the following metric over the signal set $\clH$  
\begin{align}
d_{\bfH, \bfH'}&={\|\bfH-\bfH'\|_\sfF}={\sqrt{\|\bfH\|_\sfF^2 + \|\bfH'\|_\sfF^2 - 2\inp{\bfH}{\bfH'} }}\nonumber\\
&=\sqrt{m\|\bfy\|^2 + m\|\bfy'\|^2 -2m\nu_{\bfS,\bfS'} \inp{\bfy}{\bfy'}},\label{eq:metric_H}
\end{align} 
where we used $\|\bfH\|_\sfF^2=m\|\bfy\|^2$, $\|\bfH'\|_\sfF^2=m\|\bfy'\|^2$, and that
\begin{align*}
\inp{\bfH}{\bfH'}&=\tr(\bfH^\transp \bfH')=\tr\Big ( (\bfy^\transp \otimes \bfS)(\bfy'\otimes {\bfS'}^\transp) \Big )\\
&=\inp{\bfy}{\bfy'} \inp{\bfS}{\bfS'},
\end{align*}
where we also defined the \textit{similarity metric} between  selection matrices $\bfS$ and $\bfS'$ by
\begin{align}\label{nu_def}
\nu_{\bfS,\bfS'}=\frac{\inp{\bfS}{\bfS'}}{m}=\frac{\tr(\bfS^\transp \bfS')}{m} \in [0,1],
\end{align}
which gives the fraction of similar rows in $\bfS$ and $\bfS'$. 
We  define the \textit{Relaxed} RIP over $\clH-\clH$ up to a precision $\ripc$ as follows.
\begin{definition}[R-RIP over $\clH-\clH$]\label{RIP_diff2}
Let $\bfB$ be a  matrix and let $\bbb=\vec(\bfB)$. The linear map $\bfH\mapsto \bfH\bbb$ induced by $\bbb$ satisfies the \textit{Relaxed RIP} (R-RIP)  over $\clH-\clH$ with a parameter $\ripb$ and a precision $\ripc$ if 
\begin{align*}
(1-\ripb)\|\bfH-\bfH'\|_\sfF^2 \leq \|(\bfH-\bfH')\bbb\|^2 \leq (1+\ripb)\|\bfH-\bfH'\|_\sfF^2,
\end{align*}
holds for all $\bfH, \bfH'$, with $d_{\bfH, \bfH'} \geq \ripc \max\{\|\bfH\|_\sfF, \|\bfH'\|_\sfF\}$. \hfill $\lozenge$
\end{definition}

The next proposition shows that, in a suitable regime of parameters $m,n,k$, and $\ripb, \ripc$, a matrix $\bfB$ with $\sfN(0,1)$ elements satisfies R-RIP over $\clH-\clH$  with a high probability. 
\begin{proposition}\label{rip_2}
Let $\bfB$ be a random matrix with i.i.d. $\sfN(0,1)$ elements and let $\bbb=\vec(\bfB)$. There is a constant $c>0$ such that $\bfB$ satisfies  R-RIP on $\clH-\clH$ with parameters $\ripb, \ripc$ with a probability larger than $1- 2 (1+\frac{2}{\ripb})^{2k} {n \choose m}^2 e^{-c m \ripb^2\ripc^2}$.  \hfill $\square$
\end{proposition}

\begin{proof}
	Proof in Appendix \ref{app:rip_2}. 
\end{proof}

\begin{remark}\label{eps_remark}
By taking the union bound over the concentration results derived in Proposition \ref{rip_1} and \ref{rip_2} and using the fact that the concentration in Proposition \ref{rip_1} is much sharper, we will assume in the rest of the paper that $\bfB$ satisfies both RIP over $\clH$ with a parameter $\ripa$ and R-RIP over $\clH-\clH$ with parameters $\ripb, \ripc$, with a probability approximately given by $1- 2 (1+\frac{2}{\ripb})^{2k} {n \choose m}^2 e^{-c m \ripb^2\ripc^2}$. In particular, $\ripa$ can always be taken much lower than $\ripb$ ($\ripa\ll \ripb$) without affecting this bound.  \hfill $\lozenge$ 
\end{remark}

For  a suitably selected set of parameters $n,m,k,$ and $\ripb, \ripc$, $\bfB$ satisfies R-RIP with a high probability if $2 (1+\frac{2}{\ripb})^{2k} {n \choose m}^2 e^{-c m \ripb^2\ripc^2}$ is small. Asymptotically as $m,n,k \to \infty$, this  is satisfied provided that 
\begin{align}\label{rec_cond}
\frac{k}{n} \log(1+\frac{2}{\ripb}) +\frac{1}{n} \log{n \choose m}- \frac{c m  \ripb^2 \ripc^2 }{2n} <0.
\end{align}
In particular, if ${\frac{m}{n} \to \rho}$ and ${\frac{k}{n} \to \kapp}$, applying the Stirling's approximation for ${n \choose m}$ \cite{dubhashi2009concentration}, the condition \eqref{rec_cond} takes the form 
\begin{align}\label{eq:succ_cond}
\kapp \log(1+\frac{2}{\ripb}) + h(\theta) - \frac{c \ripb^2 \ripc^2}{2} (1-\theta) <0,
\end{align}
where ${\theta:=1-\rho}=\frac{n-m}{n}$ denotes the fractional sampling loss and where ${h(\theta)=-\theta \log(\theta) - (1-\theta) \log(1-\theta)}$ denotes the entropy function (computed with the natural logarithm).
For a given  R-RIP parameters $\ripb, \ripc$, \eqref{eq:succ_cond} is satisfied for a sufficiently small $\theta$ and $\kapp$ at the cost of incurring a large oversampling factor $\frac{n}{k}=\frac{1}{\kapp}$. When only $o(n)$ number of samples are missing, i.e., $m=n-o(n)$, thus, ${\theta \to 0}$,  an oversampling of the order 
\begin{align}\label{ovs_scale}
\frac{n}{k}\approx \frac{2\log(1+\frac{2}{\ripb})}{c\ripb^2 \ripc^2}
\end{align}
is sufficient to compensate for the missing labels.

\subsection{Guarantee for Signal Recovery}\label{sec:sig_rec}
Let $\check{\bfH}=\check{\bfy}^\transp \otimes \check{\bfS}\in \clH$ and let $\check{\bfx}=\check{\bfS} \bfB  \check{\bfy}+\bfw=\check{\bfH} \bbb +\bfw$ be the set of $m$ noisy linear measurements  taken via $\bfB$, where $\bfw$ denotes the measurement noise. 
We define the measurement SNR by $\snr =\frac{\|\check{\bfH}\bbb\|^2}{\|\bfw\|^2}=\frac{\|\check{\bfS}\bfB\check{\bfy}\|^2}{\|\bfw\|^2}$ as before. 
We consider the following recovery problem
\begin{align}\label{algor}
\widehat{\bfH}=\argmin_{\bfH \in \clH} \|\bfH \bbb - \check{\bfx}\|.
\end{align}
\begin{theorem}\label{main_thm}
	Let $\bfB$ be a random  matrix that  satisfies RIP over $\clH$ with a parameter $\ripa$ and R-RIP over $\clH-\clH$ with parameters $\ripb, \ripc$. Let $\check{\bfH}$ and $\check{\bfx}$ be as before and let $\widehat{\bfH}$ be an estimate of $\check{\bfH}$ obtained  from \eqref{algor}. Then, $\|\widehat{\bfH}- \check{\bfH}\|_\sfF \leq \chi\|\check{\bfH}\|_\sfF$, where $\chi=\max \big \{ \ripc\frac{\sqrt{1+\ripa}}{\sqrt{1-\ripa}} (1+\frac{2}{\sqrt{\snr}}) ,  \frac{2}{\sqrt{\snr(1-2\ripb)}} \big \}$. \hfill $\square$
\end{theorem}
\begin{proof}
	Since $\widehat{\bfH}$ is the solution of \eqref{algor} and $\check{\bfH}$ itself is feasible, we have that 
	\begin{align}
	\|\widehat{\bfH} \bbb -\check{\bfx}\| \leq \|\check{\bfH} \bbb -\check{\bfx}\|=\|\bfw\|.
	\end{align}
	From triangle inequality, {$\|(\widehat{\bfH} - \check{\bfH}) \bbb\| \leq \|\bfw\| + \|\bfw\|=2\|\bfw\|$}, thus, $\|\widehat{\bfH} \bbb\| \leq \| \check{\bfH} \bbb\| + 2 \|\bfw\|= \| \check{\bfH} \bbb\|(1+ \frac{2}{\sqrt{\snr}})$. This yields 
	\begin{align}
	\|\widehat{\bfH}\|_{\sfF} &\stackrel{(i)}{\leq} \frac{\|\widehat{\bfH} \bbb\|}{\sqrt{1-\ripa}} \leq  \frac{\|\check{\bfH}\bbb\|}{\sqrt{1-\ripa}} (1+\frac{2}{\sqrt{\snr}})\\
	&\stackrel{(ii)}{\leq} \|\check{\bfH}\|_\sfF \frac{\sqrt{1+\ripa}}{\sqrt{1-\ripa}} (1+\frac{2}{\sqrt{\snr}}),\label{rec_dumm_1}
	\end{align}
	where in $(i)$ and $(ii)$ we applied  RIP over $\clH$ to $\widehat{\bfH}$ and $\check{\bfH}$ respectively. From \eqref{rec_dumm_1}, it results that 
	\begin{align}
	\max\{\|\widehat{\bfH}\|_{\sfF}, \|\check{\bfH}\|_{\sfF}\} \leq \|\check{\bfH}\|_\sfF \frac{\sqrt{1+\ripa}}{\sqrt{1-\ripa}} (1+\frac{2}{\sqrt{\snr}}).\label{rec_dumm_2}
	\end{align}
	If $d_{\widehat{\bfH}, \check{\bfH}} \leq \ripc \max\{\|\widehat{\bfH}\|_\sfF, \|\check{\bfH}\|_\sfF\}$,  \eqref{rec_dumm_2} yields
	\begin{align}
	\frac{\|\widehat{\bfH}- \check{\bfH}\|_\sfF}{ \|\check{\bfH}\|_\sfF}=\frac{d_{\widehat{\bfH}, \check{\bfH}}}{ \|\check{\bfH}\|_\sfF} \leq \ripc \frac{\sqrt{1+\ripa}}{\sqrt{1-\ripa}} (1+\frac{2}{\sqrt{\snr}}).\label{rec_dumm_3}
	\end{align}
	Otherwise, if $d_{\widehat{\bfH}, \check{\bfH}} \geq \ripc \max\{\|\widehat{\bfH}\|_\sfF, \|\check{\bfH}\|_\sfF\}$, R-RIP over $\clH-\clH$ holds for $\widehat{\bfH}$ and $\check{\bfH}$, where we obtain 
	\begin{align*}
	\sqrt{1-\ripb}\frac{\|\check{\bfH}-\widehat{\bfH}\|_\sfF}{\|\check{\bfH}\|_\sfF}& \stackrel{(a)}{\leq} \frac{\|(\check{\bfH}-\widehat{\bfH})\bbb\|}{\|\check{\bfH}\|_\sfF} \leq \frac{2\|\bfw\|}{\|\check{\bfH}\|_\sfF}\\
	&\stackrel{(b)}{\leq} \frac{2\|\bfw\|}{\|\check{\bfH} \bbb\| \sqrt{1-\ripa}}=\frac{2}{\sqrt{\snr(1-\ripa)}},
	\end{align*}
	where in $(a)$ we applied R-RIP to $\check{\bfH}-\widehat{\bfH}$, and where in $(b)$ we used the RIP over $\clH$ for $\check{\bfH}$. 
	After simplification, we have  
	\begin{align}
	\frac{\|\check{\bfH}-\widehat{\bfH}\|_\sfF}{\|\check{\bfH}\|_\sfF} \leq  \frac{2}{\sqrt{\snr(1-\ripb)(1-\ripa)}} \stackrel{(i)}{\leq} \frac{2}{\sqrt{\snr(1-2\ripb)}},\label{rec_dumm_4}
	\end{align}
	where in $(i)$ we used $\ripa \ll \ripb$ (see Remark \ref{eps_remark}). Combining \eqref{rec_dumm_3} and \eqref{rec_dumm_4} completes the proof.
\end{proof}

\begin{remark}\label{feas_remark}
Theorem \ref{main_thm} continues to hold for any other estimate $\widehat{\bfH}$ that merely satisfies the feasibility condition 
\begin{align}\label{feas_cond}
\widehat{\bfH} \in \{\bfH: \|\bfH \bbb - \check{\bfx}\| \leq \eta \|\bfw\|\}, 
\end{align}
for any $\eta\geq 1$, where it yields $\|\check{\bfH}-\widehat{\bfH}\|_\sfF \leq  \chi \|\check{\bfH}\|_\sfF$, with a $\chi=\max \big \{ \ripc\frac{\sqrt{1+\ripa}}{\sqrt{1-\ripa}} (1+\frac{\eta+1}{\sqrt{\snr}}) ,  \frac{\eta+1}{\sqrt{\snr(1-2\ripb)}} \big \}$. \hfill $\lozenge$
\end{remark}

Theorem \ref{main_thm} provides a recovery guarantee for any estimate $\widehat{\bfH}$ obtained from \eqref{algor} or \eqref{feas_cond} when $\bfB$ satisfies  R-RIP. An implicit requirement, however, is to have an efficient low-complexity algorithm to find such an estimate $\widehat{\bfH}$. 
In the next section, we design such a low-complexity algorithm via \textit{Alternating Minimization}.

\section{Recovery Algorithm}\label{sec:rec_alg}

Let ${\check{\bfy}}$ be a signal and let ${\check{\bfx}}={\check{\bfS}} \bfB {\check{\bfy}} +\bfw$ be the vector of noisy unlabeled measurements where $\check{\bfS} \in \clS$. We define the following cost function for the recovery of ${\check{\bfy}}$ as in \eqref{algor}:
\begin{align}
f(\bfS, \bfy)=\|{\check{\bfx}}-\bfS \bfB \bfy\|^2, \ \bfS \in \clS, \ \bfy \in \clY,\label{cost_func}
\end{align}
where ${\clY:=\bR^k}$. 
For a fixed $\bfB$ and i.i.d. Gaussian noise $\bfw$, the minimizer of $f(\bfS, \bfy)$ denoted by  $(\bfS^*,\bfy^*)$ gives the maximum likelihood (ML) estimate of $({\check{\bfS}},{\check{\bfy}})$ and consequently the ML estimate $\bfy^*$ of the desired signal ${\check{\bfy}}$ \cite{casella2002statistical}. Finding the ML estimate $\bfy^*$, however, requires a joint optimization over $(\bfS,\bfy)$. This seems to require searching over all ${n \choose m}$ possible $\bfS\in\clS$, which might be unfeasible for large $n$ and $m$. Here, instead of joint search over $\bfS$ and $\bfy$, we use an iterative alternating minimization with respect to $\bfy$ and $\bfS$  to reduce the complexity.  
Our proposed algorithm alternates between estimating a suitable selection matrix in $\bfS$ and a target signal $\bfy$ (\textit{Alternating Minimization}) and resembles the \textit{Iterative Hard  Thresholding} (IHT) algorithm proposed for signal recovery in CS \cite{blumensath2009iterative}, which also alternates between finding a suitable support $\bfL$ and a suitable vector  of nonzero coefficients $\bfx$ with the notation introduced here (see Fig.\,\ref{fig:dual_fig}). This emphasizes further the underlying duality between \namesh\ and CS.
We call our algorithm \algname.
We initialize \algname\  with a random $\bfS^1 \in \clS$ and define $(\bfS^t,\bfy^t)$ as the estimate of $({\check{\bfS}},\check{\bfy})$ obtained by \algname\ at iteration $t=1,2, \dots$, via the following sequential projection operations
\begin{align}
\bfS^t&\mapsto \bfy^t=\argmin_{\bfy \in\clY} f(\bfS^t,\bfy),\label{y_upd}\\
\bfy^t &\mapsto \bfS^{t+1}=\argmin_{\bfS \in \clS} f(\bfS,\bfy^{t}),\label{S_upd}
\end{align}
where we used ${\bfS \mapsto \bfy}$ for $\bfy=\argmin_{\bfy \in \clY} f(\bfS, \bfy)$ and $\bfy \mapsto \bfS$ for $\bfS=\argmin_{\bfS \in \clS} f(\bfS, \bfy)$ for the projection operators onto $\clY$ and $\clS$ respectively. 

\subsection{Projection on $\clY$} 
For a fixed $\bfS^t$, finding $\bfy^t$ given by ${\bfS^t \mapsto \bfy^t}$ in \eqref{y_upd} boils down to obtaining the least-squares solution of an over-complete set of linear equations (via the matrix $\bfS^t \bfB$). The optimal solution is given by  $\bfy^t=(\bfS^t \bfB)^\dagger \check{\bfx}$, where $^\dagger$ denotes the pseudo-inverse operator (for the tall matrix $\bfS^t \bfB$).

\subsection{Projection on $\clS$}
Let $\bfy^t$ be the optimal solution obtained from \eqref{y_upd} and set $\bfz^t=\bfB \bfy^t$. Finding the optimal selection matrix $\bfS^{t+1} \in \clS$ at \eqref{S_upd} given by $\bfy^t \mapsto \bfS^{t+1}$ requires extracting a subvector of $\bfz^t$ of dimension $m$, while keeping the relative order of the components, that is closest to $m$-dim vector  ${\check{\bfx}}$ in $l_2$ distance. We formulate this as a  Dynamic Programming as follows. We define a 2D table of size $m\times n$ whose elements are labeled with $(r,c)\in [m]\times [n]$ and have the value $\bfT_{r,c}\in \bR_+ \cup \{\infty\}$ given by 
\begin{align*}
\bfT_{r,c}=\left \{ \begin{array}{lr} \parbox{5.3cm}{minimum squared-distance between the subvector ${\check{\bfx}}_{\range{1}{r}}$  and a subsequence of $\bfz^t_{\range{1}{c}}$ of length $r$,} & c\geq r, \\ 
\infty, & \text{otherwise.}
\end{array} \right .
\end{align*}
We initialize the diagonal elements of the table with $\bfT_{i,i}=\|{\check{\bfx}}_{\range{1}{i}}-\bfz^t_{\range{1}{i}}\|^2$ since there is only one way to match the first $i$ elements of  $\bfz^t$ with the $i$ elements in ${\check{\bfx}}_{\range{1}{i}}$. We also initialize the elements in the first column of the table, i.e., $\bfT_{1,j}$ for $j\in [n]$, with $\bfT_{1,j}=\min_{j' \in [j]} |{\check{\bfx}}_1-\bfz^t_{j'}|^2$ since the single element ${\check{\bfx}}_{1}$ should be matched with the closest element in the subvector $\bfz^t_{\range{1}{j}}$ consisting of the first $j$ elements of $\bfz^t$.

To find the value of a generic element $\bfT_{r,c}$ in the table, we need to match ${\check{\bfx}}_{\range{1}{r}}$ with a suitable subsequence of $\bfz^t_{\range{1}{c}}$ of length $r$. In the optimal matching, the last component ${\check{\bfx}}_{r}$ is matched either with $\bfz^t_{c}$ or with $\bfz^t_{c'}$ for some $c'<c$. Thus, $\bfT_{r,c}$ can be computed from the already computed elements of the table as follows: 
\begin{align}\label{T_rec}
\bfT_{r,c}=\min\big \{&\bfT_{r-1,c-1}+|{\check{\bfx}}_r-\bfz^t_{c}|^2, \bfT_{r,c-1}\big \}.
\end{align}
With the already mentioned initialization and the recursion \eqref{T_rec}, we can complete the whole table by filling its $d$-th diagonal consisting of elements $\{\bfT_{1,d},\bfT_{2,d+1}, \bfT_{3,d+2}, \dots \}$ for $d=2,\dots, n$ one at a time. Overall, filling the whole table requires $O(mn)$ operations.
\begin{figure}[t]
\centering
\includegraphics{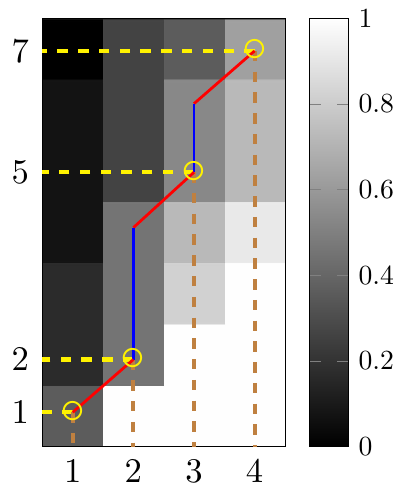}
\caption{Illustration of the Dynamic Programming table for matching two sequences of length $7$ and $4$ respectively. For the illustrated example, the indices of the matched elements in the larger vector are given by $\{1,2,5,7\}$.}
\label{fig:dynamic_prog}
\end{figure}

After filling the whole table, we can find the indices of those $m$ elements of $\bfz^t$ that are optimally matched to the elements of ${\check{\bfx}}$ as follows. We start from the element $\bfT_{m,n}$ located at the up-right corner of the table at location $(m,n)$. Note that by definition, $i \mapsto \bfT_{m,i}$, for $i \in [n]$, is a decreasing sequence of $i$ since by increasing $i$ the subvector $\bfz^t_{\range{1}{i}}$ becomes longer and provides more options to find a better subsequence thereof matching ${\check{\bfx}}$. The index of the last element in $\bfz^t$ that is matched to the last element ${\check{\bfx}}_{m}$ of ${\check{\bfx}}$ in the optimal matching is 
\begin{align}
i_m=\min \big \{i \in [n]: \bfT_{m,i}=\bfT_{m,n}\big \}.
\end{align}
To find the next largest index $i_{m-1}$, we apply the same argument to the sub-table $\bfT_{\range{1}{m-1}, \range{1}{i_m-1}}$ and its up-right corner element located at $(m-1,i_m-1)$, where we obtain the following recursive formula for the remaining indices:
\begin{align*}
i_{m-l}=\min \big\{&i \in [i_{m-l+1}-1]: \bfT_{m-l,i}=\bfT_{m-l,i_{m-l+1}-1}\big \},
\end{align*}
where ${l \in [m-1]}$. 
Fig.\,\ref{fig:dynamic_prog} illustrates this for a vector $\bfz$ of dimension $7$ and a vector $\bfx$ of dimension $4$.

\section{Performance Analysis}
In this section, we analyze the performance of \algname\ under the assumption that the measurement matrix $\bfB$ satisfies  R-RIP over $\clH-\clH$. Intuitively speaking, we show that  \algname\ can be seen as an iterative procedure for finding  fixed points of an appropriate  set-valued map. We use this  to analyze the performance of \algname\ rigorously. 

\subsection{Set-valued Map}\label{sec:set_val}

Let $\clX$ and $\clZ$ be two arbitrary sets and let $\clF\subseteq \clX\times \clZ:=\{(\bfx,\bfz): \bfx \in \clX, \bfz \in \clZ\}$. We define the set-valued or the multi-valued map corresponding to $\clF$ by $\clF: \clX \to 2^\clZ$, where $2^\clZ$ denotes the power set of $\clZ$, and where $\clF$ assigns to each $\bfx\in \clX$ a subset of $\clZ$ given by
\begin{align}
\clF(\bfx)=\{\bfz \in \clZ: (\bfx,\bfz) \in \clF\}.
\end{align}
Note that, for simplicity, we use the same symbol $\clF$ both for a subset of $\clX\times \clZ$ and for the associated set-valued map. In the special case, where $\clF(\bfx)$ is a singleton (has only a single element) for all $\bfx\in \clX$, the set-valued map $\clF$ corresponds to a single-valued function $\clF: \clX \to \clZ$. We will use $\clF: \clX \rightrightarrows \clZ$ for a set-valued map $\clF$. We define the domain of $\clF$ by $\clD_{\clF}:=\{\bfx \in \clX: \clF(\bfx) \neq \emptyset\}$ and the range of $\clF$ by $\clR_{\clF}=\cup_{\bfx \in \clX} \clF(\bfx)$. In particular, $\bfz \in \clR_\clF$ if there is an $\bfx\in \clX$ such that $\bfz \in \clF(\bfx)$ or equivalently $(\bfx,\bfz) \in \clF$. 
We define the inverse of a set-valued map $\clF$ as a mapping $\clF^{-1}: \clZ \rightrightarrows \clX$, where $\bfx \in \clF^{-1}(\bfz)$ if and only if $\bfz \in \clF(\bfx)$. A simple example of a set-valued map is illustrated in Fig.\,\ref{fig:set_valued}.
\begin{figure}[t]
\centering
\includegraphics{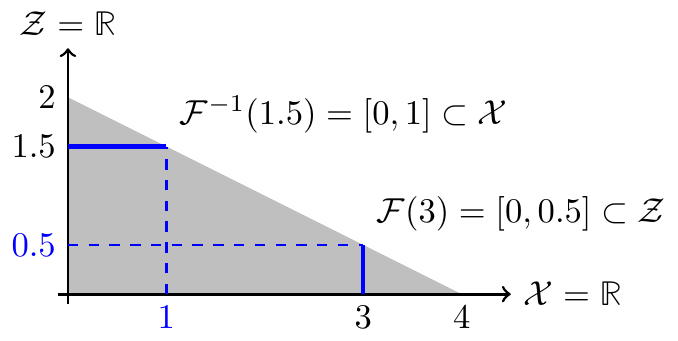}
\caption{Illustration of a simple set-valued map $\clF\subset \clX \times \clZ$ and its inverse for $\clX=\clZ=\bR$.}
\label{fig:set_valued}
\end{figure}

We define the composition of two set-valued maps  $\clF: \clX \rightrightarrows \clZ$ and $\clG: \clZ \rightrightarrows \clW$  by a set-valued map $\clJ=\clG\circ \clF: \clX \rightrightarrows \clW$, where $\bfw \in \clJ(\bfx)$ if and only if there is a $\bfz \in \clZ$ connecting $\bfx$ and $\bfw$, i.e.,  $\bfz \in \clF(\bfx)$ and $\bfw\in \clG(\bfz)$. This can also be written more compactly as $\clJ(\bfx)=\cup_{\bfz \in \clF(\bfx)} \clG(\bfz)$. We refer to \cite{rockafellar2015convex, dontchev2009implicit} for a more comprehensive introduction to the set-valued maps and additional related references.

\subsection{Decoding up to a Radius}\label{sec:dec_radius}

As in  Section \ref{sec:rec_alg}, we denote the target signal by $\check{\bfH}=\check{\bfy}^\transp \otimes \check{\bfS}$, the noisy samples by $\check{\bfx}=\check{\bfH} \bbb+\bfw=\check{\bfS} \bfB  \check{\bfy} +\bfw$, and the sequence generated by \algname\ by $\bfS^1 \mapsto \bfy^1 \mapsto \bfS^2 \mapsto \cdots$. 
Let $\bfH^i={\bfy^i}^\transp \otimes \bfS^i$ be the estimates produced by \algname\ at the end of the $i$-th iteration.
Since \algname\ minimizes the cost function \eqref{cost_func} at each iteration, we have that $\|\bfH^i \bbb - \check{\bfx}\| \leq \|\check{\bfx}\|$. Applying the triangle inequality and using the RIP over $\clH$ for $\check{\bfH}$ and $\bfH^i$ (as in the proof of Theorem \ref{main_thm}) yields 
\begin{align}\label{max_norm_cond}
\|\bfH^i\|_\sfF \leq \frac{2}{\sqrt{1-2\ripa}}(1+\frac{1}{\sqrt{\snr}}) \|\check{\bfH}\|_\sfF \stackrel{(i)}{\approx} 2 \|\check{\bfH}\|_\sfF,
\end{align}
where $(i)$ holds for a large $\snr$ and a small $\ripa$. 
Let us consider the following condition for $(\bfH^i, \check{\bfH})$:
\begin{align}
(1-\ripb)d_{\bfH^i,\check{\bfH}}^2 \leq \|(\bfH^i-\check{\bfH})\bbb\|^2 \leq (1+\ripb)d_{\bfH^i,\check{\bfH}}^2,\label{R_cond}
\end{align}
where $d_{\bfH^i,\check{\bfH}}=\|\bfH^i-\check{\bfH}\|_{\sfF}$ as in \eqref{eq:metric_H}.
Since $\bfB$ satisfies R-RIP over $\clH-\clH$, if \eqref{R_cond} is violated then from R-RIP it results that
\begin{align}
\|\bfH^i- \check{\bfH}\|_\sfF \leq \mu \max \{ \|\bfH^i\|_\sfF, \|\check{\bfH}\|_\sfF\} \stackrel{(a)}{\approx} 2 \mu \|\check{\bfH}\|_\sfF
\end{align}
where  $(a)$ follows from \eqref{max_norm_cond}. In words, the violation of \eqref{R_cond} implies that the estimate $\bfH^i$  lies in a neighborhood of $\check{\bfH}$ of radius $2 \ripc \|\check{\bfH}\|_\sfF$, as illustrated in Fig.\,\ref{fig:R_RIP}.
For the analysis,  we assume that \eqref{R_cond} holds for all $(\bfH^i, \check{\bfH})$ and study \algname\ under this condition. In particular, letting $\|\bfH^i-\check{\bfH}\|^{\text{Analysis}}_\sfF$ be the performance obtained under this additional condition, we have
\begin{align}
\|\bfH^i-\check{\bfH}\|^{\text{\algname}}_\sfF \leq \max\big\{2 \ripc \|\check{\bfH}\|_\sfF , \|\bfH^i-\check{\bfH}\|^{\text{Analysis}}_\sfF\big\},\label{eq:dec_radius}
\end{align} 
where $\|\bfH^i-\check{\bfH}\|^{\text{\algname}}_\sfF$ denotes the true performance achieved by \algname\ without the additional condition \eqref{R_cond}.

\subsection{Alternating Minimization as a Set-valued Map}
We  define the normalized cost function $g(\bfS,\bfy):=\frac{\sqrt{f(\bfS, \bfy)}}{\sqrt{m} \|\check{\bfy}\|}$, where $f(\bfS, \bfy)$ is as in \eqref{cost_func}. We  introduce the normalized variables $\xi=\frac{\|\bfy\|^2}{\|\check{\bfy}\|^2}$ and $\vartheta=\frac{\inp{\bfy}{\check{\bfy}}}{\|\check{\bfy}\|^2}$, and ${\nu:=\nu_{\bfS, \check{\bfS}}}$, where ${\nu_{\bfS, \check{\bfS}}=\frac{\tr(\bfS^\transp \check{\bfS})}{m}}$ as in \eqref{nu_def} denotes the fraction of common rows in $\bfS$ and $\check{\bfS}$.
As $\bfy$ varies in ${\clY=\bR^k}$, the set of feasible $(\xi, \vartheta)$ is given by
\begin{align}\label{eq:eset}
\clE=\{(\xi, \vartheta): \xi\geq \vartheta^2\},
\end{align}
which is the area above the parabola $\xi=\vtheta^2$  (see Fig.\,\ref{fig:feasible_area}).

Applying the triangle and the reverse triangle inequality to $\sqrt{f(\bfS, \bfy)}=\|\bfS \bfB \bfy - \check{\bfS} \bfB  \check{\bfy} -\bfw\|$ and using \eqref{R_cond} for ${\bfH=\bfy^\transp \otimes \bfS}$ and an arbitrary ${\check{\bfH}=\check{\bfy}^\transp \otimes \check{\bfS}}$, we obtain the following upper and lower bounds for  $g(\bfS,\bfy)$:
\begin{align}
g(\bfS,\bfy)& \leq \sqrt{1+\ripb}\, \varphi(\xi,\vartheta,\nu) + \zeta,\label{upper_bb}\\
g(\bfS,\bfy)& \geq \sqrt{1-\ripb}\, \varphi(\xi,\vartheta,\nu) - \zeta, \label{lower_bb}
\end{align}
where {$\zeta:=\frac{\|\bfw\|}{\|\check{\bfy}\|\sqrt{m}}\leq \frac{\|\bfw\|}{\|\check{\bfS}\bfB\check{\bfy}\|\sqrt{1-\ripa}}=\frac{1}{\sqrt{\snr(1-\ripa)}}$}, where $\ripa$ is the parameter of RIP over $\clH$, and where
\begin{align}
\varphi(\xi,\vartheta,\nu)&=\frac{\|\bfH-\check{\bfH}\|_\sfF}{\sqrt{m} \|\check{\bfy}\|}=\frac{d_{\bfH,\check{\bfH}}}{\sqrt{m} \|\check{\bfy}\|}\nonumber\\
&\stackrel{(i)}{=}\sqrt{1+\frac{\|\bfy\|^2}{\|\check{\bfy}\|^2} - 2 \nu \frac{\inp{\bfy}{\check{\bfy}}}{\|\check{\bfy}\|^2}}\label{eta_func2}\\
&=\sqrt{1+\xi -2 \vartheta \nu}, \label{eta_func}
\end{align}
where in $(i)$ we used \eqref{eq:metric_H}, and where $(\xi, \vtheta)$ takes values in $\clE$ as in \eqref{eq:eset}.
With this notation, we investigate the properties of the solution path $\bfS^1 \mapsto \bfy^1 \mapsto \cdots \mapsto {\bfS^t \mapsto \bfy^t} \mapsto \bfS^{t+1} \cdots$ generated by \algname.

\begin{figure}[t]
	\centering
	\includegraphics[scale=0.8]{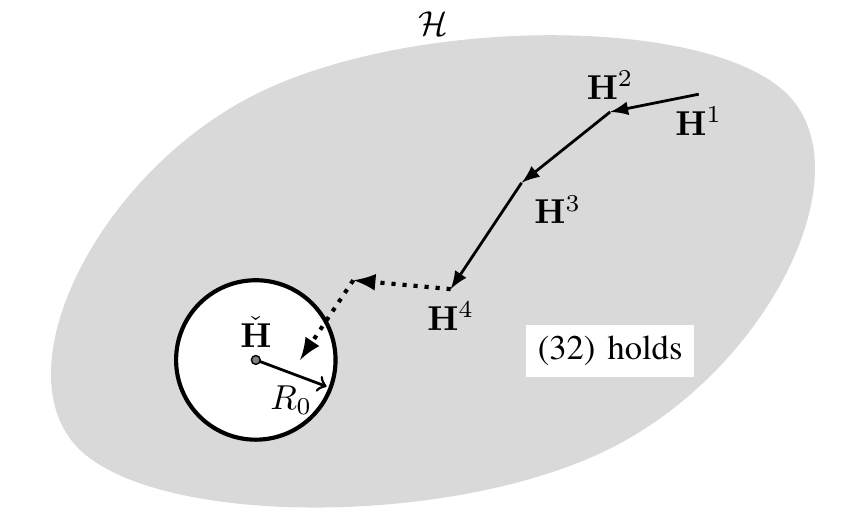}
	\caption{The path of \algname\ over $\clH$ and the region of $\clH$ where \eqref{R_cond} holds for $\check{\bfH}$ ($R_0\approx 2\ripc \|\check{\bfH}\|_\sfF$).}
	\label{fig:R_RIP}
\end{figure}

\subsubsection{Projection on $\clY$}
Consider an iteration $t$, where \algname\  produces  $\bfy^t$ as the projection ${\bfS^t \mapsto \bfy^t}$ of $\bfS^t$ onto $\clY$. From the upper bound in \eqref{upper_bb}, we have that
\begin{align}
g(\bfS^t,\bfy^t)&=\argmin_{\bfy \in \clY} g(\bfS^t,\bfy)\nonumber\\
& \leq \argmin_{(\xi, \vartheta) \in \clE} \sqrt{1+\ripb}\, \varphi(\xi,\vartheta,\nu^t) + \zeta\nonumber\\
&= \sqrt{1+\ripb} \sqrt{1-(\nu^t)^2} +\zeta, \label{nu_der_dumm}
\end{align}
where ${\nu^t=\nu_{\check{\bfS}, \bfS^t}}$. 
To derive \eqref{nu_der_dumm}, we used the fact that $\varphi(\xi, \vartheta, \nu^t)$ in \eqref{eta_func} is a decreasing function of $\xi$ for a fixed $\vtheta$ and $\nu^t \in [0,1]$, thus, the minimum of $\varphi(\xi, \vartheta, \nu^t)$ over $\clE$ is achieved at the boundary ${\{(\xi, \vtheta): \xi=\vtheta^2\}}$ of $\clE$ (see Fig.\,\ref{fig:feasible_area}). After replacing  $\xi=\vtheta^2$ in $\varphi$, we need to minimize $\sqrt{1+\vtheta^2-2\nu^t \vtheta}$ with respect to $\vtheta \in \bR$, where the minimum is achieved  at $\vtheta=\nu^t$ and its value is given by $\sqrt{1-(\nu^t)^2}$ as in \eqref{nu_der_dumm}. We could obtain this result also directly by minimizing \eqref{eta_func2} with respect to $\bfy$ (rather than $(\xi, \vtheta)\in \clE$), where we obtain the optimal solution $\bfy=\nu^t \check{\bfy}$ and \eqref{nu_der_dumm} after replacement.

Let us denote by $\xi^t=\frac{\|\bfy^t\|^2}{\|\check{\bfy}\|^2}$ and $\vtheta^t=\frac{\inp{\bfy^t}{\check{\bfy}}}{\|\check{\bfy}\|^2}$ the  variables corresponding to $\bfy^t$. Using the lower bound \eqref{lower_bb}, we have $g(\bfS^t,\bfy^t) \geq \sqrt{1-\ripb}\, \varphi(\xi^t, \vtheta^t,\nu^t) - \zeta$, which with \eqref{nu_der_dumm} yields  
\begin{align}\label{nu_val}
\varphi(\xi^t,\vartheta^t,\nu^t)&=\sqrt{1+{\xi^t} -2 \vtheta^t \nu^t}\nonumber\\
& \leq \varsigma \sqrt{1-(\nu^t)^2} +\varrho=:\vpi(\nu^t),
\end{align}
where $\varsigma:=\sqrt{\frac{1+\ripb}{1-\ripb}}$, where $\varrho:=\frac{2\zeta}{\sqrt{1-\ripb}}$, and where we defined 
\begin{align}\label{eq:vpi}
\vpi(\nu):=\varsigma \sqrt{1-\nu^2} +\varrho.
\end{align}
We write  \eqref{nu_val} as  ${(\xi^t, \vtheta^t) \in\Fxi(\nu^t)}$, where $\Fxi(\nu)$ is the set-valued map defined by
\begin{align}
\Fxi(\nu):=\big \{(\xi, \vartheta):  \xi \geq \vtheta^2, \sqrt{1+{\xi} -2 \vartheta \nu}\leq \vpi(\nu)\big\},\label{F_xi}
\end{align}
where $\vpi(\nu)$ is as in \eqref{eq:vpi}. Note that $\Fxi(\nu)$ 
assigns to any $\nu \in [0,1]$ a subset of feasible points $\clE$ that can be \textit{potentially} produced by \algname.
Before we proceed we make the following  crucial condition.
\begin{condition}\label{condition1}
	There is a $\bar{\nu} \in (0,1]$ such that for $\nu \in ({\bar{\nu}, 1]}$
	\begin{align}\label{vpi_cond}
	\vpi(\nu,\vsigma,\vrho)=\varsigma \sqrt{1-\nu^2} +\varrho <  1,
	\end{align} 
	where $\vsigma=\sqrt{\frac{1+\ripb}{1-\ripb}}$ as before. Such a $\bar{\nu}$ 
	exists if
	\begin{align}\label{vrho_cond}
	\vrho=&\frac{2\zeta}{\sqrt{1-\ripb}}\stackrel{(i)}{\leq}\frac{2}{\sqrt{\snr(1-\ripb)(1-\ripa)}}< 1,
	\end{align}
	where  $(i)$ follows from {$\zeta\leq\frac{1}{\sqrt{\snr(1-\ripa)}}$}. We define $\nu_0$ as the smallest such $\bar{\nu}$, which from \eqref{vpi_cond} is given by 
	\begin{align}
	\nu_0:=&\sqrt{1-(\frac{1-\vrho}{\vsigma})^2}.\label{nu0_cond}
	\end{align}
\end{condition}
Condition \ref{condition1} is satisfied for a sufficiently large $\snr$ and a sufficiently small $\ripb \in (0,1)$, and is easy to meet in practice. 
\begin{proposition}\label{pos_cor}
Let $\nu_0$ be as in {Condition \ref{condition1}} and let ${\nu \in (\nu_0,1]}$. Then, for any $(\xi, \vtheta)\in\Fxi(\nu)$,  $\vtheta\geq 0$. \hfill $\square$
\end{proposition} 
\begin{proof}
From \eqref{eq:vpi}, it is seen that $\vpi(\nu)$ is a decreasing function of $\nu$, thus, from ${\nu \in (\nu_0,1]}$ it immediately results that ${\vpi(\nu) < \vpi(\nu_0) =1}$. From \eqref{F_xi}, this implies that 
$\sqrt{1+\xi -2 \vtheta \nu} <1$. 
 Since $\nu>0$ and $\xi\geq 0$, this  is satisfied only if $\vtheta\geq 0$. This completes the proof.
\end{proof}

\begin{figure}[t]
	\centering
	\includegraphics{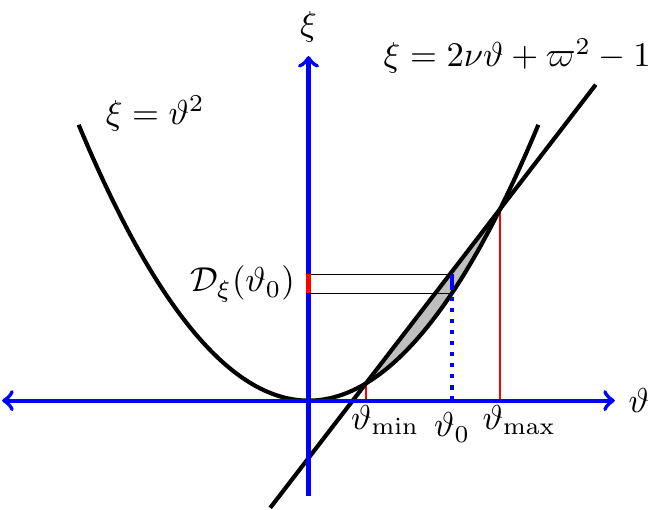}
	\caption{Illustration of  the domain $\Fxi(\nu)$ for a generic $\nu\in  (\nu_0,1]$. It is seen that $\Fxi(\nu)$ consists of the area between the parabola $\xi=\vtheta^2$ and the line $\xi=2\nu \vtheta + \vpi^2-1$ in $\vtheta\in [\vtheta_{\min}, \vtheta_{\max}]$. Also, all $(\xi, \vtheta)\in\Fxi(\nu)$ have positive $\vtheta$. }
	\label{fig:feasible_area}
\end{figure}

\subsubsection{Projection on $\clS$}
Let $(\xi^t, \vtheta^t)$ be the output of  \algname\  corresponding to $\bfy^t$. The next step is the projection $\bfy^t \mapsto \bfS^{t+1}$ onto $\clS$, where we have 
\begin{align}
g(\bfS^{t+1},\bfy^t)&=\argmin_{\bfS\in \clS} g(\bfS,\bfy^t)\\
& \stackrel{(i)}{\leq} \argmin_{\nu \in [0,1]} \sqrt{1+\ripb}\ \varphi(\xi^t,\vartheta^t,\nu) + \zeta\\
&\stackrel{(ii)}{=} \argmin_{\nu\in [0,1]} \sqrt{1+\ripb}\ \sqrt{1+{\xi^t} -2 \nu \vtheta^t} + \zeta\\
&\stackrel{(iii)}{\leq}\sqrt{1+\ripb}\ \sqrt{1+{\xi^t} -2 (\vtheta^t)_+} + \zeta,\label{up_bound_dumm}
\end{align}
where in $(i)$ we used the upper bound \eqref{upper_bb}, where in $(ii)$ we used \eqref{eta_func}, where we defined $(a)_+=\max\{a, 0\}$ as the positive part of $a\in \bR$, and where in $(iii)$ we used the fact that when $\vtheta^t>0$ the minimum is achieved at $\nu=1$ whereas for $\vtheta^t \leq 0$ it is achieved at $\nu=0$.  
Using \eqref{up_bound_dumm} and the lower bound \eqref{lower_bb} at $(\bfS^{t+1}, \bfy^t)$, we obtain 
\begin{align}\label{nu_eq_up}
\sqrt{1+{\xi^t} -2 \vtheta^t \nu^{t+1}} \leq \varsigma\sqrt{1+{\xi^t} -2 (\vtheta^t)_+} +\varrho.
\end{align}
Although \eqref{nu_eq_up} does not give the exact value of $\nu^{t+1}$ generated by \algname, it implies  that ${\nu^{t+1} \in \Fnu(\xi^t, \vtheta^t)}$, where $\Fnu(\xi, \vtheta)$ is the set-valued map defined by 
\begin{align}
\Fnu(\xi, \vtheta):=\big \{ \nu \in [0,1]: &\sqrt{1+{\xi} -2 \vtheta \nu} \nonumber\\
&\leq \varsigma \sqrt{1+{\xi} -2 (\vtheta)_+} +\varrho\big \}.\label{F_nu}
\end{align}

\subsubsection{Full Iteration}\label{sec:full_iter}
Combining the two steps of projection in \algname\  at  iteration $t$, we obtain that starting from $\nu^t \in (\nu_0,1]$, the algorithm returns $\nu^{t+1} \in \clF(\nu^t)$, where $\clF$ is the composition of  the set-valued maps $\Fxi$ and $\Fnu$ given by
\begin{align}\label{f_both}
\clF(\nu):=&\Fnu \circ \Fxi (\nu)\nonumber\\
=&\{\nu' \in\Fxi(\xi', \vtheta'): (\xi', \vtheta')\in \Fnu(\nu)\}.
\end{align}
In words, $\clF(\nu)$ is the set of all $\nu'$ in $[0,1]$ that are \textit{potentially} reachable from $\nu$ at a single iteration of \algname.

\subsection{Structure of $\clF(\nu)$ and its lower envelope}
We define the lower envelope of the set-valued map $\clF$ as the single-valued function
\begin{align}
\sfF_0(\nu):=\min \clF(\nu)=\min \{\nu': \nu' \in \clF(\nu)\},\label{F}
\end{align}
where $\sfF_0(\nu) \in [0,1]$ for all $\nu \in [0,1]$. Let $\nu_0$ be as in Condition \ref{condition1}. We will study the behavior of $\sfF_0$ in two regions  ${\nu \in (\nu_0,1]}$ and $\nu\in [0, \nu_0]$ separately. 

\subsubsection{First case: ${\nu \in (\nu_0,1]}$}
For this range of $\nu$, from Proposition \ref{pos_cor} it results that  for any $(\xi, \vtheta)\in\Fxi(\nu)$, we have $\vtheta \geq 0$, thus, $(\vtheta)_+=\vtheta$. For each such a $(\xi, \vtheta) \in\Fxi(\nu)$ with $\vtheta\geq 0$, \eqref{F_nu} yields a lower bound on the possible values of $\nu$ as follows
\begin{align}\label{F_def}
\nu \geq \max \Big \{ \frac{1+\xi - (\vsigma \sqrt{1+\xi - 2 \vtheta}+ \vrho)^2}{2\vtheta}, 0\Big \}.
\end{align}
However, to find a lower envelope of $\clF(\nu)$, from \eqref{f_both} we need to consider all those $(\xi, \vtheta)\in\Fxi(\nu)$.
Thus, the lower envelope  $\sfF_0(\nu)$ is given by $\sfF_0(\nu)=\max \{ \sfF(\nu), 0\}$, where
\begin{align}\label{F_def}
\sfF(\nu)=\min_{(\xi, \vtheta)\in\Fxi(\nu)}\frac{1+\xi - (\vsigma \sqrt{1+\xi - 2 \vtheta}+ \vrho)^2}{2\vtheta}.
\end{align}
The region $\Fxi(\nu)$ is illustrated in Fig.\,\ref{fig:feasible_area} for a generic value of $\nu\in(\nu_0,1]$.
It is seen that $\Fxi(\nu)$ for any fixed $\nu\in(\nu_0,1]$ consist of all $(\xi, \vtheta)$ lying above the parabola $\xi=\vtheta^2$ and below the line $\xi=2\nu \vtheta + \vpi^2-1$ in the range $[\vtheta_{\min}, \vtheta_{\max}]$, where $\vtheta_{\min}$ and $\vtheta_{\max}$ denote the $\vtheta$ coordinate of the two intersection points of the line and the parabola and correspond to the roots of the  quadratic equation obtained by setting $\xi=\vtheta^2$ in the equation of the line $\xi=2\nu \vtheta + \vpi^2-1$ given by:
\begin{align}
\sfQ(\vtheta):=\vtheta^2 - 2 \vtheta \nu+1-\vpi^2=0.\label{quad_eq}
\end{align} 
Solving \eqref{quad_eq} we have 
\begin{align}\label{roots}
\vtheta_{{\max/\min}}(\nu, \vsigma, \vrho)=\nu \pm \sqrt{\nu^2 -1 +\vpi(\nu)^2 },
\end{align}
where $\vpi(\nu)$ is as in \eqref{eq:vpi}.
Note that the roots are well-defined (real-valued) for all $\nu \in [0,1]$ and in particular for ${\nu \in (\nu_0,1]}$ because the discriminant 
\begin{align}\label{Delt_rel}
\Delta(\nu,&\vsigma, \vrho)=\nu^2-1+\vpi(\nu)^2=\nu^2-1+(\vsigma \sqrt{1-\nu^2}+\rho)^2 \nonumber\\
&= (\vsigma^2-1) (1-\nu^2) + \rho^2 + 2\rho\vsigma \sqrt{1-\nu^2}\geq0
\end{align}
since $\vsigma=\sqrt{\frac{1+\ripb}{1-\ripb}}>1$ and $1-\nu^2 \geq 0$ for $\nu \in [0,1]$. From the root relation for a quadratic function, we also have that  
\begin{align}\label{root_rel}
\vtheta_{\min}(\nu, \vsigma, \vrho) \vtheta_{\max}(\nu, \vsigma, \vrho) = 1-\vpi^2 \in (0,1],
\end{align}
for ${\nu \in (\nu_0,1]}$ as {$\vpi(\nu) \in [0,1)$} for ${\nu \in (\nu_0,1]}$ from Condition \ref{condition1}. As ${{\vtheta_{\max}>0}}$, \eqref{root_rel} also yields ${\vtheta_{\min}>0}$, thus, both roots are positive for ${\nu \in (\nu_0,1]}$. Moreover, as ${\vtheta_{\min} \vtheta_{\max}\leq 1}$, the minimum root $\vtheta_{\min}$ is always less than $1$.
\subsubsection{Second case: $\nu \in[0,\nu_0]$}
In this region, we see from Condition \ref{condition1} and \eqref{eq:vpi} that $\vpi(\nu) \geq 1$, thus, $(\xi, \vtheta)=(0,0) \in\Fxi(\nu)$ from \eqref{F_xi}. We can see from \eqref{F_nu} and the fact that $\vsigma \in (1, \infty)$ and $\vrho \in [0,1)$ from \eqref{vrho_cond} that for $(\xi, \vtheta)=(0,0)$,  $\Fnu(\xi, \vtheta)=[0,1]$. From the composition of set-valued maps  in \eqref{f_both}, it results that $\clF(\nu)=[0,1]$ for  $\nu \in [0,\nu_0]$, and in particular, $\sfF_0(\nu)=0$.

\subsection{Explicit formula for $\sfF_0(\nu)$}
Suppose ${\nu \in (\nu_0,1]}$ and  write the minimization \eqref{F_def} as 
\begin{align}
\sfF(\nu)&=1+ \min_{(\xi, \vtheta)\in\Fxi(\nu)}\frac{1+\xi -2 \vtheta - (\vsigma \sqrt{1+\xi - 2 \vtheta}+ \vrho)^2}{2\vtheta}\nonumber\\
&=1+\min_{(\xi, \vtheta)\in\Fxi(\nu)} \frac{\sfG(1+\xi - 2\vtheta)}{2\vtheta},\label{F_def2}
\end{align}
where $\sfG: a \to a-(\vsigma \sqrt{a} + \vrho)^2$ is a function defined over the region $\clD_\sfG:=\{1+\xi-2\vtheta: (\xi, \vtheta) \in\Fxi (\nu)\} \subset \bR$. Note that as $\xi \geq \vtheta^2$ over the region $\Fxi(\nu)$, we have that
\begin{align}
\min \clD_\sfG\geq \min \{1+\vtheta^2-2\vtheta: \vtheta \in [\vtheta_{\min}, \vtheta_{\max}]\} \geq 0.
\end{align}
Hence, $\clD_\sfG \subseteq \bR_+$ and $\sfG$ is well-defined over the whole  $\clD_\sfG$. Moreover, since $\clD_\sfG$ is the image of the closed connected (and bounded) set $\Fxi(\nu)$ (see Fig.\,\ref{fig:feasible_area}) under the continuous affine map $(\xi, \vtheta) \mapsto 1+\xi - 2\vtheta$, the set $\clD_\sfG$ must be closed and connected \cite{bertsekas2003convex}, thus, a closed interval of $\bR_+$. 
\begin{lemma}\label{G_lem}
 $\sfG$ is a negative decreasing function over $\bR_+$ and in particular $\clD_G$. \hfill $\square$
\end{lemma}
\begin{proof}
By expanding the expression for $\sfG$, we have that ${\sfG(a)=(1-\vsigma^2) a -2\vsigma \vrho \sqrt{a} - \vrho^2}$. As $\vsigma >1$ and $a\in \bR_+$ (especially when $a \in \clD_G$), it results that $\sfG(a) \leq 0$ for all $a \in \bR_+$ ($a \in \clD_G$). Also, $\frac{d}{d a}\sfG(a)= \sfG'(a)=(1-\vsigma^2) -\frac{\vsigma \vrho}{\sqrt{a}}<0$ for $a \in \bR_+$ ($a\in \clD_G$). Thus, $\sfG$ is a decreasing function.
\end{proof}
\begin{lemma}\label{corner}
Let ${\nu \in (\nu_0,1]}$. Then, the minimum in \eqref{F_def} over $\Fxi(\nu)$ is achieved at  corner point $(\xi, \vtheta)=(\vtheta_{\min}, \vtheta_{\min}^2)$. \hfill $\square$
\end{lemma}
\begin{proof}
We first fix a $\vtheta>0$ and consider $\sfG(1+\xi-2\vtheta)$ as a function of $\xi$ in the range $\clD_{\text{\textxi}}(\vtheta):=\{\xi: (\xi, \vtheta)\in\Fxi(\nu)\}$, which is a closed interval (see Fig.\,\ref{fig:feasible_area}). As $1+\xi-2\vtheta$ for a fixed $\vtheta$ is an affine increasing function of $\xi$, from Lemma \ref{G_lem} it results that $\sfG(1+\xi-2\vtheta)$ is a decreasing function of $\xi$ in $\clD_{\text{\textxi}}(\vtheta)$, thus, its minimum over $\clD_{\text{\textxi}}(\vtheta)$ is achieved at $\xi_{\max}(\vtheta)=\max \clD_{\text{\textxi}}(\vtheta)$. As a result,  $(\xi_{\max}(\vtheta), \vtheta)$ lies on the boundary line $1+\xi-2\nu \vtheta=\vpi^2$, thus, $\xi_{\max}(\vtheta)=\vpi^2-1+2\nu \vtheta$. Replacing $\xi_{\max}(\vtheta)$ for $\xi$, the minimum in \eqref{F_def2} is given by
\begin{align}
\sfF(\nu)=1+ \min _{\vtheta \in [\vtheta_{\min}, \vtheta_{\max}]} \frac{\sfG(\vpi^2+ 2 \vtheta (\nu-1))}{2\vtheta}.
\end{align}
Since $\vpi^2+ 2 \vtheta (\nu-1)$ is an affine decreasing function of $\vtheta$ for ${\nu \in (\nu_0,1]}$, we have that 
\begin{align}\label{pos_G_der}
{\frac{d}{d\vtheta} \sfG(\vpi^2+ 2 \vtheta (\nu-1))=2(\nu-1) \sfG'(\vpi^2+ 2 \vtheta (\nu-1)) \geq 0},
\end{align}
where we used the fact that $\nu-1\leq 0$ and that $\sfG'(.)\leq 0$ from Lemma \ref{G_lem}. This implies that 
\begin{align*}
&\frac{d}{d\vtheta} \frac{\sfG(\vpi^2+ 2 \vtheta (\nu-1))}{2\vtheta}\\
&=\frac{\frac{d}{d\vtheta} \sfG(\vpi^2+ 2 \vtheta (\nu-1)) \vtheta -  \sfG(\vpi^2+ 2 \vtheta (\nu-1))}{2 \vtheta^2}\geq 0
\end{align*}
where we used \eqref{pos_G_der}, the fact that $\vtheta>0$, and that $\sfG(\vpi^2+ 2 \vtheta (\nu-1))\leq 0$ from Lemma \ref{G_lem}.  As a result, $\frac{\sfG(\vpi^2+ 2 \vtheta (\nu-1))}{2\vtheta}$ is an increasing function of $\vtheta$ for $\vtheta \in [\vtheta_{\min}, \vtheta_{\max}]$ and achieves its minimum at $\vtheta_{\min}$. Overall, this implies that for ${\nu \in (\nu_0,1]}$ the minimum in \eqref{F_def2} is achieved at the corner point $(\vtheta_{\min}^2, \vtheta_{\min}) \in\Fxi(\nu)$. This completes the proof.
\end{proof}

Applying Lemma \ref{corner}, we obtain the following explicit formula for $\sfF(\nu)$ for $\nu \in (\nu_0, 1]$:
\begin{align}
\sfF(\nu)=1+ \frac{(1-\vtheta_{\min})^2 -(\vsigma(1-\vtheta_{\min})+\vrho)^2}{2 \vtheta_{\min}},\label{thet_rel}
\end{align}
where  $\vtheta_{\min}$ is as in \eqref{roots} and where we  used the fact that ${\sqrt{1+\vtheta_{\min}^2 -2\vtheta_{\min}}=1-\vtheta_{\min}}$ as ${\vtheta_{\min} \in [0,1]}$ for ${\nu \in (\nu_0,1]}$. 
Thus, we obtain  
\begin{align}\label{F0_exp}
\sfF_0(\nu)=\left \{ \begin{array}{ll} 0, & \nu \in [0, \nu_0], \\ \big(\sfF(\nu))_+, & {\nu \in (\nu_0,1]}.
\end{array} \right.
\end{align}

\subsection{Properties of $\sfF_0(\nu)$}
We need  some  preliminary results first.
\begin{lemma}\label{thet_min_inc}
Let $\vtheta_{\min}(\nu)$ be as in \eqref{roots}. Then, $\vtheta_{\min}(\nu)$ is an increasing function of $\nu$ for  ${\nu \in (\nu_0,1]}$. Moreover, $\vtheta_{\min} (\nu_0)=0$ and $\vtheta_{\min}(1)=1-\vrho$. \hfill $\square$
\end{lemma}
\begin{proof}
From Condition \ref{condition1}, $\vpi(\nu_0)=1$.  Thus, $\vtheta_{\min}(\nu_0)=\nu_0 - \sqrt{\nu_0^2 -1 +\vpi(\nu_0)^2 }=0$. Also, replacing $\nu=1$ and using the fact that ${\vpi(\nu)=\vsigma \sqrt{1-\nu^2}+\vrho=\vrho}$ at $\nu=1$ yields ${\vtheta_{\min}(1)=1-\vrho}$. Replacing $\vpi(\nu)=\vsigma \sqrt{1-\nu^2}+\vrho$ in $\vtheta_{\min}(\nu)$, we can write 
\begin{align}\label{thmin_dumm}
\vtheta_{\min}(\nu)&=\nu-\sqrt{-(1-\nu^2)+(\vsigma \sqrt{1-\nu^2}+\vrho)^2}\\
&=\nu - \sqrt{-\sfG(1-\nu^2)},
\end{align}
where ${\sfG: a \to a-(\vsigma \sqrt{a} + \vrho)^2}$ is as before. From Lemma \ref{G_lem}, it results that $-\sfG(.)$ is a positive increasing  function over $\bR_+$ and in particular $[0,1]$. Thus, $-\sqrt{-\sfG(.)}$ is a decreasing and $-\sqrt{-\sfG(1-\nu^2)}$ is an increasing function of $\nu$. Hence,  $\vtheta_{\min}(\nu)=\nu - \sqrt{-\sfG(1-\nu^2)}$ is an increasing function. 
\end{proof}

We first consider \eqref{thet_rel} and write $\sfF(\nu)=\sfM\circ \vtheta_{\min}(\nu)$, where ${\sfM: \bR_+ \to \bR}$ is defined by
\begin{align}
\sfM(a)&=1+\frac{(1-a)^2-(\vsigma(1-a) + \vrho)^2}{2a}\\
&=\vsigma \vrho + \vsigma^2 + \frac{1}{2} \Big ( (1-\vsigma^2) a + \frac{1-(\vsigma + \vrho)^2}{a}\Big).
\end{align}
\begin{lemma}\label{M_lem}
{$\sfM$ is an increasing  function over} $(0,1]$. \hfill $\square$
\end{lemma}
\begin{proof}
Taking the derivative of $\sfM(a)$, we have 
\begin{align}
\sfM'(a)&=\frac{1}{2} \Big ( (1-\vsigma^2)  - \frac{1-(\vsigma + \vrho)^2}{a^2}\Big).
\end{align}
Note that $\sfM'(a)$ has a  positive root at $a_0=\sqrt{\frac{(\vsigma+\vrho)^2-1}{\vsigma^2-1}}\geq1$ as $\vsigma>1$ and $\vrho\geq0$. Moreover, $\sfM'(0^+)=+\infty$ and $\sfM'(a)\geq 0$ for $a\in (0, a_0]$ and in particular for all $a\in(0,1]$. This implies that $\sfM(a)$ is an increasing function of $a$ for $a\in(0,1]$. 
\end{proof}
\begin{proposition}\label{F_mon_con}
$\sfF_0(\nu)$ is an increasing function over $(\nu_0,1]$.
\end{proposition}
\begin{proof}
We first prove that $\sfF(\nu)$ is an increasing function of $\nu$.
Note that $\sfF(\nu)=\sfM\circ \vtheta_{\min}(\nu)$, where for ${\nu \in (\nu_0,1]}$, $\vtheta_{\min}(\nu)\in (0,1]$, over which $\sfM$ is an increasing function from Lemma \ref{M_lem}. Moreover, $\vtheta_{\min}(\nu)$ is also an increasing function of $\nu$ for ${\nu \in (\nu_0,1]}$ from Lemma \ref{thet_min_inc}. This implies the composition function $\sfF(\nu)=\sfM\circ\vtheta_{\min}(\nu)$ is an increasing function of $\nu$ over $\nu\in (\nu_0,1]$. From \eqref{F0_exp} and the fact that $a \to (a)_+$ is an increasing function of $a$, we obtain that $\sfF_0(\nu)$ is an increasing function of $\nu$ over $[0,1]$. This completes the proof.
\end{proof}

Proposition \ref{F_mon_con} implies that
$\sfF_0(\nu)\in [0, (\sfF_{\max})_+]$, where
\begin{align}
\sfF_{\max}(\vsigma, \vrho):=&\sfM(\vtheta_{\min}(1))\stackrel{(a)}{=}\sfM(1-\vrho)\nonumber\\
=&1+ \frac{1-(1+\vsigma)^2}{2} \frac{\vrho^2}{1-\vrho},\label{F_max}
\end{align} 
where in $(a)$ we used the fact that $\vtheta_{\min}(1)=1-\vrho$ from Lemma \ref{thet_min_inc}. 
From \eqref{F_max}, it is seen that $\sfF_{\max}\leq 1$ as $\vsigma>1$ and $\vrho \in [0,1)$. Moreover, $\sfF_{\max}>0$ when $\vrho$ is sufficiently small (large $\snr$) and $\vsigma$ is not far from $1$ (small R-RIP parameter $\ripb$), and  $\lim _{{(\vsigma, \vrho)\to (1,0)}} \sfF_{\max}(\vsigma, \vrho)= 1$. Since $\sfF(\nu)$ is an  increasing function of $\nu$, provided that $\sfF_{\max}(\vsigma, \vrho)>0$ there would exist a ${\nu_1(\vsigma, \vrho) \in (\nu_0,1]}$ such that $\sfF(\nu)>0$, thus, $\sfF_0(\nu)>0$ for $\nu \in (\nu_1,1]$. 
A direct calculation by setting $\sfF(\nu_1)=\sfM\circ \vtheta_{\min}(\nu_1)=0$ yields
\begin{align}
\vtheta_{\min}(\nu_1(\vsigma, \vrho))=\frac{(\vsigma+\vrho)^2-1}{\vsigma^2 + \vsigma \vrho + \sqrt{\vsigma^2+ (\vsigma+\vrho)^2 -1}},
\end{align}
which using $\vtheta_{\min}(\nu)=\nu -\sqrt{\nu^2-1+(\vsigma \sqrt{1-\nu^2}+\vrho)^2}$ can be solved to obtain the value of $\nu_1(\vsigma, \vrho)$ explicitly.

\subsection{Fixed points of $\sf\sfF_0(\nu)$}
The crucial step in our analysis is based on the fixed points of $\sfF_0(\nu)$ defined by $\{\nu\in [\nu_1,1]: \sfF_0(\nu)=\nu\}$. 
Note that the possible fixed points of $\sfF_0(\nu)$  in the range $\nu \in [0,1]$ are also fixed points of $\sfF(\nu)$. Form \eqref{thet_rel}, these fixed points, provided that they exist, should satisfy the following equation 
\begin{align}
\nu=1+ {\frac{(1-\vtheta_{\min}(\nu))^2 -(\vsigma(1-\vtheta_{\min}(\nu))+\vrho)^2}{2 \vtheta_{\min}(\nu)}}.
\end{align}
A straightforward calculation yields
\begin{align}
\vtheta_{\min}^2 +1 -2 \nu \vtheta_{\min} -(\vsigma(1-\vtheta_{\min}(\nu))+\vrho)^2=0.
\end{align}
Using the identity ${\sfQ(\vtheta_{\min})=\vtheta_{\min}^2 - 2 \vtheta_{\min} \nu+1-\vpi^2=0}$ as in \eqref{quad_eq}  and doing some simplification results in  
\begin{align}
\vsigma \sqrt{1-\nu^2} + \vrho=\vpi(\nu)=\vsigma(1-\vtheta_{\min}(\nu)) + \vrho,
\end{align}
which can be  simplified to $\vtheta_{\min}(\nu)=1-\sqrt{1-\nu^2}$ and  written, using \eqref{eq:vpi} and \eqref{roots}, more explicitly in terms of $\nu$ as follows 
\begin{align*}
\nu-\sqrt{\nu^2-1+(\vsigma \sqrt{1-\nu^2}+\vrho)^2}=\vtheta_{\min}(\nu)=1-\sqrt{1-\nu^2}.
\end{align*}
By introducing the auxiliary variable ${\nu=\sin(\alpha)}$ for $\alpha \in [0,\frac{\pi}{2}]$, we can write this equivalently as 
\begin{align}\label{eq:fx_pt}
\big(\vsigma \cos(\alpha) + \vrho\big)^2 - \cos(\alpha)^2=\big(1-\cos(\alpha)-\sin(\alpha) \big)^2.
\end{align}
We will consider the noiseless ($\vrho=0$) and the noisy ($\vrho \in (0,1)$) cases separately.

\subsubsection{Noiseless Case  ($\vrho=0$)}
In the noiseless case, using the fact that ${\cos(\alpha)+\sin(\alpha) \geq 1}$ for $\alpha\in [0, \frac{\pi}{2}]$, \eqref{eq:fx_pt} yields 
\begin{align}\label{eq:fix1}
\sin(\alpha) + (1-\sqrt{\vsigma^2-1}) \cos(\alpha)=1.
\end{align}
By introducing  $\tan(\beta)=1-\sqrt{\vsigma^2-1}$ for $\beta \in (-\frac{\pi}{2}, \frac{\pi}{2})$, we can write \eqref{eq:fix1} as $\sin(\alpha+\beta)=\cos(\beta)$. One of the solutions is $\alpha_{\max}=\frac{\pi}{2}$ and corresponds to the largest fixed point $\nu_{\max}=\sin(\frac{\pi}{2})=1$. The second solution is given by $\alpha+\beta=\frac{\pi}{2} - \beta$, or $\alpha_{\min}=\frac{\pi}{2} - 2\beta$ and lies in the allowed range $[0, \frac{\pi}{2}]$ provided that $\beta \in [0, \frac{\pi}{4}]$ or equivalently $\tan(\beta) \in [0,1]$. This restricts the range of  permitted $\vsigma$ to $\vsigma \in (1, \sqrt{2}]$. From $\vsigma=\sqrt{\frac{1+\ripb}{1-\ripb}}$, this yields the bound $\ripb \in (0, \frac{1}{3})$ on  R-RIP parameter $\ripb$. Moreover,
\begin{align}
\nu_{\min}(\vsigma)&=\sin(\alpha_{\min})=\cos(2\beta)=\frac{1-\tan(\beta)^2}{1+\tan(\beta)^2}\nonumber\\
&=\frac{1-\vsigma^2 + 2\sqrt{\vsigma^2-1}}{1+\vsigma^2+2\sqrt{\vsigma^2-1}}.\label{nu_min_noiseless}
\end{align}
It is seen that $\nu_{\min}(\vsigma) \in (0,1)$ and $\nu_{\min}(\vsigma) \to 0$ as $\vsigma \to 1$.

\subsubsection{Noisy Case ($\vrho\in (0,1)$)}
A full analysis of the possible roots of \eqref{eq:fx_pt} is more involved in the noisy case. 
We first write \eqref{eq:fx_pt} after factorizing the first term as
$\big ((\vsigma-1)\cos(\alpha)+ \vrho\big ) \big ((\vsigma-1)\cos(\alpha)+ \vrho\big ) =\big (\cos(\alpha)+\sin(\alpha)-1\big) ^2$. Since $\vsigma \in (1, \infty)$ and $\cos(\alpha) \in [0,1]$ for $\alpha \in [0, \frac{\pi}{2}]$, we have $((\vsigma\pm 1)\cos(\alpha)+ \vrho\big ) >0$. Thus, we can write \eqref{eq:fx_pt} equivalently as $\upsilon(\alpha)=0$ where 
\begin{align}
\upsilon(\alpha)={\frac{(\vsigma-1)\cos(\alpha)+ \vrho}{\cos(\alpha)+\sin(\alpha)-1}-\frac{\cos(\alpha)+\sin(\alpha)-1}{(\vsigma+1)\cos(\alpha)+ \vrho}}. \label{upsi_fx}
\end{align}
\begin{proposition}\label{two_roots}
Let $\upsilon(\alpha)$ be as in \eqref{upsi_fx}.  For any $\vsigma \in (1, \infty)$ and $\vrho \in (0, 1)$,   $\upsilon(\alpha)$ is a convex function of $\alpha$ for $\alpha \in (0, \frac{\pi}{2})$, with $\upsilon(0^+)=\upsilon(\frac{\pi}{2}^-)=+\infty$. Moreover, $\upsilon(\alpha)$ has at most two roots in $(0, \frac{\pi}{2})$. \hfill $\square$
\end{proposition}
\begin{proof}
Proof in Appendix \ref{app:two_roots}.
\end{proof}
Fig.\,\ref{fig:upsi_fx} illustrates $\upsilon(\alpha)$  for ${(\vsigma, \vrho)=(1.03, 0.06)}$, where it is seen that it has two roots. 
We will consider only those $(\vsigma, \vrho)$ in
\begin{align}\label{D_xi_rho}
\Dxirho:=\big \{(\vsigma, \vrho): \text{$\upsilon(\alpha)$ has two  roots}\big\},
\end{align}
and will denote the roots by $\alpha_{\min}$ and $\alpha_{\max}$ or by $\alpha_{\min}(\vsigma, \vrho)$ and $\alpha_{\max}(\vsigma, \vrho)$ to emphasize the implicit dependence on $(\vsigma, \vrho)$.
We also denote the corresponding fixed points of $\sfF_0$ by $\nu_{\min}=\sin(\alpha_{\min})$ and $\nu_{\max}=\sin(\alpha_{\max})$.
In view of Proposition \ref{two_roots}, a simple sufficient condition for $(\vsigma, \vrho) \in \Dxirho$ is obtained by setting $\upsilon(\frac{\pi}{4})< 0$, which can be simplified to  
\begin{align}
\vsigma + \vrho\sqrt{2} < \sqrt{7-4 \sqrt{2}}\approx 1.16,
\end{align}
which is satisfied for  $\vsigma \to 1$ (small R-RIP parameter $\ripb$) and $\vrho \to 0$ (large $\snr$).
We have the following useful result.

\begin{figure}[t]
	\centering
	\includegraphics[scale=0.9]{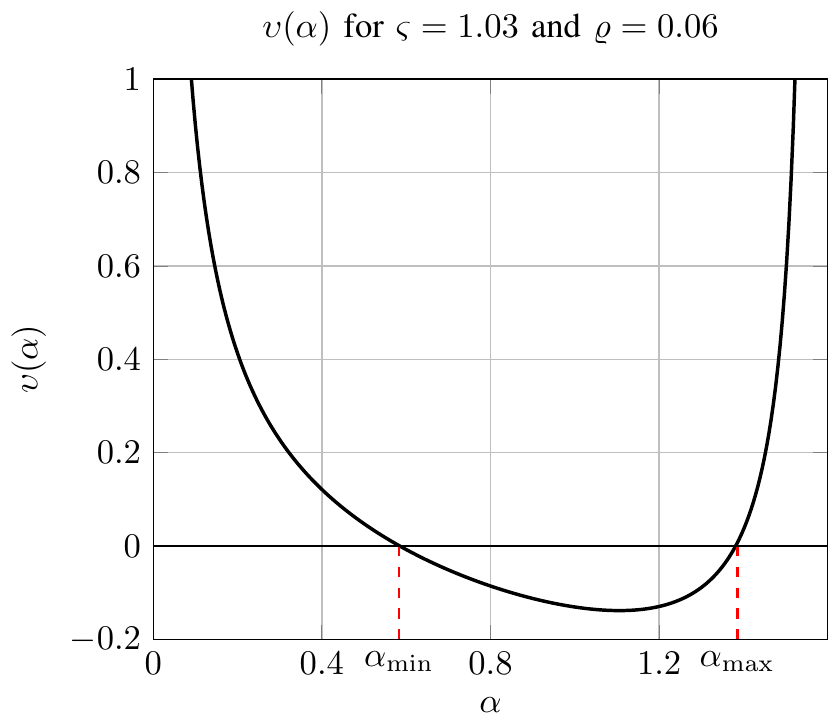}%
	\caption{Function $\upsilon(\alpha)$ for $\alpha \in (0, \frac{\pi}{2})$ and $\vsigma=1.03$, $\vrho=0.06$.}
	\label{fig:upsi_fx}
\end{figure}

\begin{proposition}\label{root_prop}
	Let $\upsilon(\alpha)$ be as in \eqref{upsi_fx} and let $\alpha_{\min}(\vsigma, \vrho)$ and $\alpha_{\max}(\vsigma, \vrho)$ be the roots of $\upsilon(\alpha)$ for $(\vsigma, \vrho) \in \Dxirho$. Then, 
	\noindent $(i)$ $\Dxirho$ is an open subset of $(1, \infty) \times (0,1)$.
	
	\noindent $(ii)$ $\alpha_{\min}$ and $\alpha_{\max}$ are differentiable functions in $\Dxirho$.
	
	\noindent $(iii)$ $\alpha_{\min}$ ($\alpha_{\max}$) is an increasing (decreasing) function of $(\vsigma, \vrho)$, i.e., $\alpha_{\min}(\vsigma', \vrho') \geq \alpha_{\min}(\vsigma, \vrho)$ for $\vsigma'\geq \vsigma$ and $\vrho'\geq \vrho$.
	
	\noindent $(iv)$ $\alpha_{\min}(\vsigma, \vrho) \to 0$ and $\alpha_{\max}(\vsigma, \vrho) \to \frac{\pi}{2}$ as ${(\vsigma, \vrho)\to (1,0)}$.
\end{proposition}
\begin{proof}
	Proof in Appendix \ref{app:root_prop}.
\end{proof}

\subsection{Evolution Equation for Alternating Minimization}
Let $\nu^t=\nu_{\check{\bfS}, \bfS^t}$ be the similarity factor of the solution $\bfS^t$ to the target $\check{\bfS}$ obtained at iteration $t$ of \algname. Although we cannot control the value of $\nu^t$, we can guarantee that $\nu^{t+1} \in \clF(\nu^t)$. In particular, this implies that $\nu^{t+1} \geq \sfF_0(\nu^t)$, where $\sfF_0(\nu)$ denotes the lower envelope of $\clF(\nu)$ as before.  
Since $\sfF_0$ is an increasing function, by repeated application of $\sfF_0$, it results that $\nu^{t+1} \geq \sfF_0^{(t)}(\nu_{\check{\bfS},\bfS^1})$, where 
\begin{align}
\sfF_0^{(t)}=\underbrace{\sfF_0 \circ \dots \circ \sfF_0}_{\text{$t$ times }}
\end{align} 
denotes the $t$-th order composition of $\sfF_0$.
\begin{proposition}[Evolution Equation]\label{EE_prop}
Let $(\vsigma, \vrho) \in \Dxirho$ and let  $\nu_{\min}(\vsigma, \vrho)$ and $\nu_{\max}(\vsigma, \vrho)$ be the two fixed points of $\sfF_0$. Let $\bfS^1$ be an initialization with $\nu^1=\nu_{\bfS^1, \check{\bfS}} > \nu_{\min}$. Then, $\liminf_{t\to \infty} \nu^t \geq \nu_{\max}$. \hfill $\square$
\end{proposition}
\begin{proof}
Applying induction and using $\nu^1 > \nu_{\min}$ and $\sfF_0(\nu_{\min})=\nu_{\min}$, we can show that $\nu^t \geq \nu_1> \nu_{\min}$. Taking the limit and denoting by $\nu_\infty=\liminf_{t \to \infty} \nu^t$, we obtain $\nu_\infty\geq \nu^1 > \nu_{\min}$, thus, $\nu_\infty > \nu_{\min}$. Moreover, we have
\begin{align}
\nu^\infty&=\liminf_{t\to \infty} \nu^{t+1}\geq \liminf_{t\to \infty} \sfF_0(\nu^{t})\nonumber\\
& \stackrel{(i)}{=} \sfF_0(\liminf_{t\to \infty} \nu^t)=\sfF_0(\nu^\infty),
\end{align}
where in $(i)$ we used the fact that $\sfF_0$ is an increasing function.
Since $\nu^\infty >\nu_{\min}$, the only region where $\nu^\infty \geq \sfF_0(\nu^\infty)$ is satisfied is when $\nu^\infty \in [\nu_{\max},1]$. This completes the proof.
\end{proof}

\begin{figure}[t]
	\centering
	\includegraphics[scale=1]{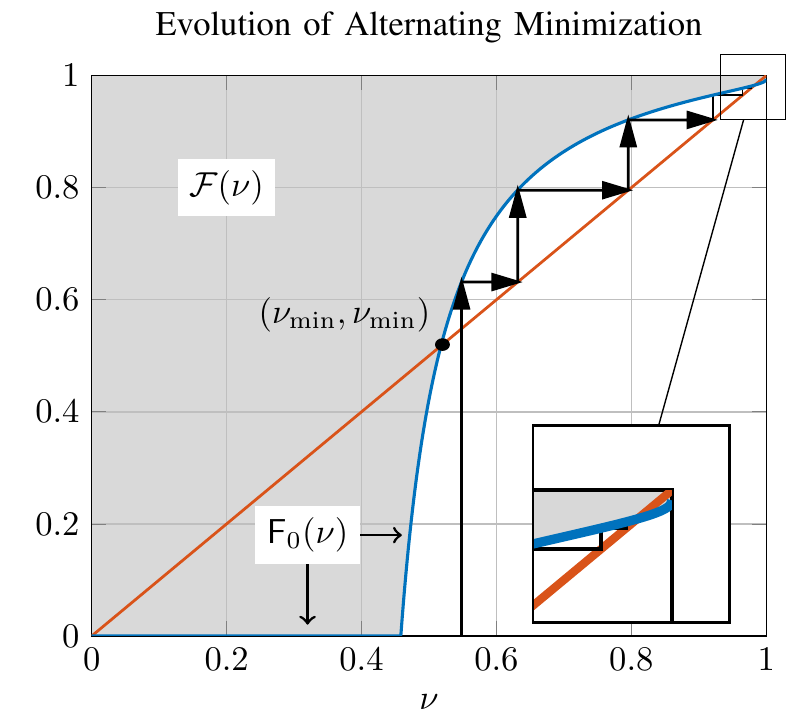}
	\caption{Illustration of the set-valued map $\clF(\nu)$, its lower envelope $\sfF_0(\nu)$, and the Evolution Equation for \algname\ starting from a ${\nu_{\check{S},S^1}> \nu_{\min}}$ for ${\vsigma=1.03}$ and ${\vrho=0.06}$.  }
	\label{fig:DE}
\end{figure}

 Fig.\,\ref{fig:DE} illustrates the set-valued map $\clF$ and its lower envelope $\sfF_0$ for $\vsigma=1.03$ and $\vrho=0.06$. It also illustrates the  evolution of $\nu^t$ produced by \algname\ across consecutive iterations. 
Using Proposition \ref{EE_prop}, we are finally in a position to analyze the performance of \algname.

\begin{theorem}[Noiseless case]\label{thm_NL}
	Let $\bfB$ be a  matrix with an R-RIP parameter $\ripb \in (0, \frac{1}{3})$, let $\vsigma=\sqrt{\frac{1+\ripb}{1-\ripb}}$, and let $\nu_{\min}(\vsigma)$ be as in \eqref{nu_min_noiseless}, where $\nu_{\min}(\vsigma) \to 0$ as ${\ripb \to 0}$ (${\vsigma \to 1}$). Let $\check{\bfx}=\check{\bfS} \bfB  \check{\bfy}$ be the unlabeled samples  taken from the signal $\check{\bfy}$ and let ${\bfS^1 \mapsto \bfy^1 \mapsto \cdots}$ be the sequence generated by \algname\ for the input $\check{\bfx}$ starting from an initialization $\bfS^1$ with ${\nu_{\bfS^1,\check{\bfS}} > \nu_{\min}(\vsigma)}$. Then, $\lim_{t \to \infty} \|\check{\bfy} - \bfy^t\|=0$. \hfill $\square$
\end{theorem}
\begin{proof}
	For the noiseless case, we have $\nu_{\max}=1$. Thus, under the stated conditions, starting from $\nu^1=\nu_{\bfS^1, \check{\bfS}}$, we have that $\liminf_{t \to \infty} \nu^t\geq\nu_{\max}=1$ from Proposition \ref{EE_prop}, which implies that $\lim_{t \to \infty}\nu^t =1$. From \eqref{nu_val} ($\vrho=0$), this yields 
	\begin{align}
	\limsup_{t \to \infty} \varphi(\xi^t, \vtheta^t, \nu^t)\leq \limsup_{t \to \infty} \vsigma \sqrt{1 - (\nu^t)^2}=0.\label{final_eq_nl}
	\end{align}
	Using {$\varphi(\xi^t, \vtheta^t, \nu^t)=\sqrt{1+{\xi^t} -2 \vtheta^t \nu^t}$} from \eqref{eta_func} and the fact that $\xi^t \geq (\vtheta^t)^2$,  \eqref{final_eq_nl} implies that $\vtheta^t \to 1$ and $\xi^t \to 1$ as $\nu^t \to 1$. This yields $\frac{\|\bfy^t\|^2}{\|\check{\bfy}\|^2} \to 1$ and $\frac{\inp{\bfy^t}{\check{\bfy}}}{\|\check{\bfy}\|^2}\to 1$, and gives the desired result $\lim_{t \to \infty} \|\bfy^t - \check{\bfy}\|=0$. 
\end{proof}
We have also the following results for the noisy case.
\begin{theorem}[Decoding up to the Noise Radius]\label{thm_NY}
Let $\bfB$ be a matrix as in Theorem \ref{thm_NL}. Suppose that $(\vsigma, \vrho)$ are in the region $\Dxirho$ where $\sfF_0(\nu)$ has two fixed points. 
Also, let $\check{\bfx}=\check{\bfS} \bfB  \check{\bfy}+\bfw$ be the noisy unlabeled samples  taken from the signal $\check{\bfy}$. Assume that ${\bfS^1 \mapsto \bfy^1 \mapsto \cdots}$ is the sequence generated by \algname\ for the input $\check{\bfx}$ starting from an initialization $\bfS^1$ with $\nu_{\check{\bfS}, \bfS^1}\geq \nu'>0$. Then, 
$ \limsup_{{(\vsigma, \vrho)\to (1,0)}} \limsup_{t \to \infty} 
\frac{\|\bfS^t \bfB \bfy^t - \check{\bfx}\|}{\|\bfw\|} \leq 3$. \hfill $\square$
\end{theorem}

\begin{figure*}[t]
\centering
\subfloat[Random initialization. \label{fig:rand_init}]{%
\includegraphics[scale=0.76]{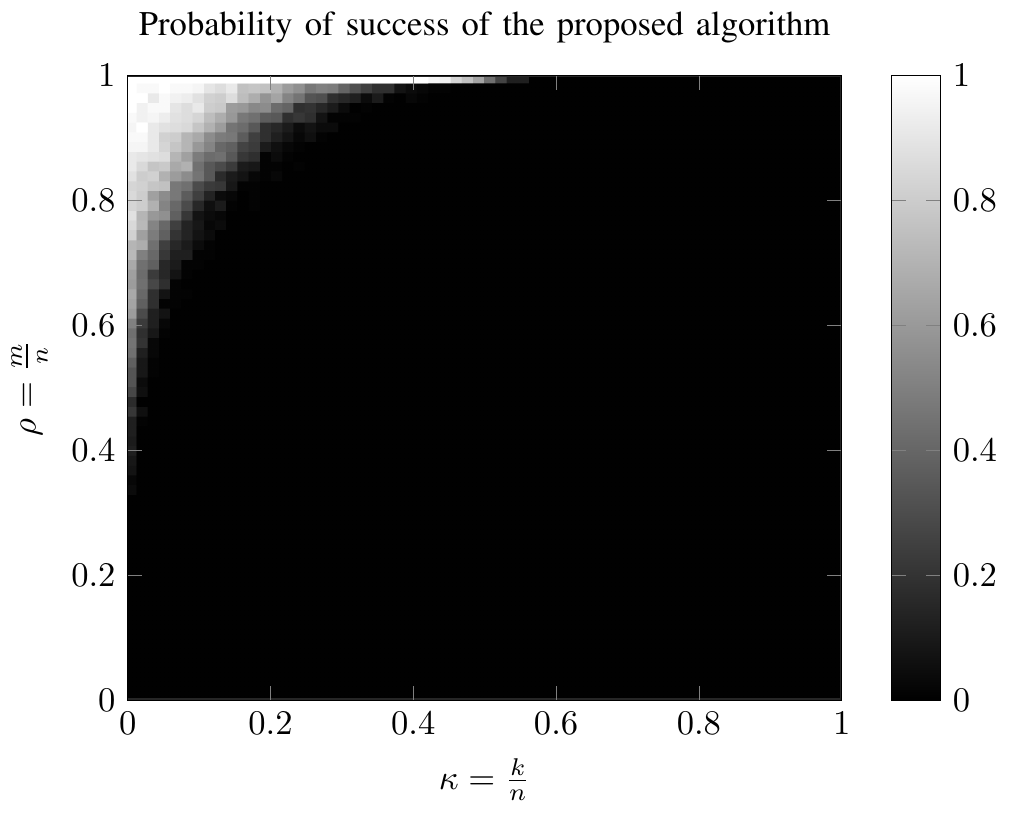}
}
\hfill
\subfloat[Genie-aided initialization. \label{fig:genie_init}]{%
\includegraphics[scale=0.76]{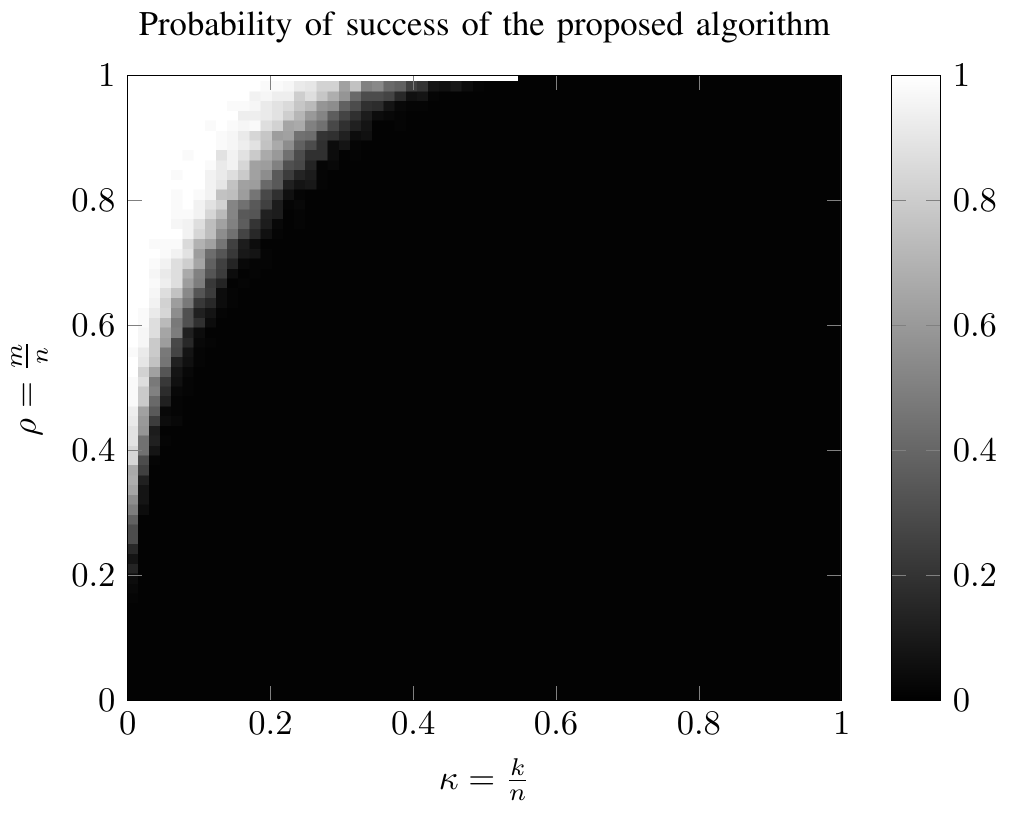}
}
\caption{Probability of  success of \algname\ as a function of {$\kapp=\frac{k}{n}$} and $\rho=\frac{m}{n}$ with a random initialization $\bfS^1$ (a) and a genie-aided initialization $\bfS^1$ with {$\nu_{\bfS^1,\check{\bfS}}=0.2$} (b), where in the latter the remaining $80\%$ of rows of $\bfS^1$ are selected randomly.}
\label{fig:prob_success}
\end{figure*}

\begin{proof}
	From proposition \ref{root_prop}, we have that $\alpha_{\min}(\vsigma, \vrho)\to 0$, thus, ${\nu_{\min}(\vsigma, \vrho) \to 0}$, as ${(\vsigma, \vrho)\to (1,0)}$. Hence, there is a sufficiently small neighborhood of $(1,0)$ such that $\nu_{\min}<\nu'$. Fixing a $(\vsigma, \vrho)$ in this neighborhood, \eqref{nu_val} yields
	\begin{align}
	\limsup_{t\to \infty} \varphi(\xi^t, \vtheta^t, \nu^t)&  \leq \limsup_{t\to \infty} \vsigma \sqrt{1-(\nu^t)^2} + \vrho\\
	&\stackrel{(i)}{\leq} \vsigma \sqrt{1-\nu_{\max}^2} + \vrho\\
	& \stackrel{(ii)}{=}\vsigma \cos(\alpha_{\max}) + \vrho,\label{phi_last_step}
	\end{align}
	where in $(i)$ we used $\liminf_{t \to \infty} \nu^t \geq \nu_{\max}$ from Proposition \ref{EE_prop}, and where in $(ii)$ we replaced ${\nu_{\max}=\sin(\alpha_{\max})}$. Using $g(\bfS^t,\bfy^t)=\frac{\sqrt{f(\bfS^t,\bfy^t)}}{\sqrt{m} \|\check{\bfy}\|}$ and  $f(\bfS^t,\bfy^t)=\|\bfS^t \bfB \bfy^t - \check{\bfx}\|^2$, we have
	\begin{align*}
	\limsup_{t \to \infty} \frac{\|\bfS^t \bfB \bfy^t -\check{\bfx}\|}{\|\bfw\|} &\stackrel{(a)}{\leq} \frac{1}{\zeta} \big ( \sqrt{1+\ripb}  \limsup_{t \to \infty} \varphi(\xi^t, \vtheta^t, \nu^t) + \zeta \big )\nonumber \\
	&\stackrel{(b)}{\leq} \sqrt{1+\ripb}\, \frac{\vsigma \cos(\alpha_{\max}) + \vrho }{\zeta} +1.
	\end{align*}
	where in $(a)$ we used  \eqref{upper_bb} and replaced $\zeta= \frac{\|\bfw\|}{\|\check{\bfy}\|\sqrt{m}}$, and where in $(b)$ we used \eqref{phi_last_step}.
	Using $\vrho=\frac{2\zeta}{\sqrt{1-\ripb}}$, taking the limit as ${(\vsigma, \vrho)\to (1,0)}$ (thus, $\ripb \to 0$), and using the fact that $\lim_{{(\vsigma, \vrho)\to (1,0)}} \alpha_{\max}(\vsigma, \vrho) = \frac{\pi}{2}$ from Proposition \ref{root_prop}, we obtain 
	\begin{align}
	\limsup_{{(\vsigma, \vrho)\to (1,0)}} \limsup_{t \to \infty} 
	\frac{\|\bfS^t \bfB \bfy^t - \check{\bfx}\|}{\|\bfw\|}\leq 3,
	\end{align}
	which is the desired result. This completes the proof.
\end{proof}

\subsection{Summary of the Analysis of \algname}
Theorem \ref{thm_NY}  proves that \algname\ is essentially able to decode the target signal up to  \textit{thrice} the noise radius  when $\snr$ is sufficiently large ($\vrho \to 0$) and the R-RIP parameter $\ripb$ is sufficiently small ($\vsigma \to 1$). 
In practice, we can afford only a finite $\snr$. Moreover, obtaining smaller $\ripb$ requires taking much more measurements scaling like $O(\frac{1}{\ripb^2})$ as in \eqref{ovs_scale}. However, in view of Remark \ref{feas_remark}, a reasonable recovery is still possible  by decoding the signal up to a multiple of noise radius $\eta \|\bfw\|$, where $\eta$ can be selected sufficiently large such that the recovery is still possible for a reasonable $\snr$ and $\ripb$. From Theorem \ref{thm_NL}, the situation is much better in the noiseless case, where $\ripb \in (0,\frac{1}{3})$ along with a \textit{good} initialization $\bfS^1$ will guarantee a suitable recovery of $\check{\bfy}$.  

Recall that for analyzing \algname, we made the assumption that  condition \eqref{R_cond} holds for all the solutions  {$\bfH^i={\bfy^i}^\transp \otimes \bfS^i$} produced by \algname. In particular, from our discussion in Section \ref{sec:dec_radius} (see also Fig.\,\ref{fig:R_RIP} and especially \eqref{eq:dec_radius}), it results that under a good initialization $\bfS^1$ \algname\ fulfills the conditions of Theorem \ref{main_thm}. In brief, \algname\ is able to recover the desired signal $\check{\bfH}$ up to a relative precision $2\ripc$ (see also Fig.\,\ref{fig:R_RIP}).

\begin{remark}\label{genie_remark}
	A key assumption in our analysis is that  $\bfB$ satisfies R-RIP over $\clH-\clH$. Our numerical simulations in Section \ref{sec:sim}, however, illustrate that \algname\ still performs quite well even in a regime of parameters  where R-RIP does not hold. Therefore,  as in CS \cite{candes2005decoding},  R-RIP seems to be sufficient but not necessary for  signal recovery. In contrast, our simulations evidently confirm that a \textit{good} initialization of $\bfS^1$ plays a crucial role on the performance of \algname, which partly validates the results we obtained in this section via a fixed-point analysis of \algname\ (even though our results were derived under R-RIP).
	\hfill $\lozenge$
\end{remark}

%

\section{Simulation Results}\label{sec:sim}
We run numerical simulation to assess the performance of \algname. For each simulation, we generate an $n\times k$ Gaussian  matrix $\bfB$ and a $k$-dim signal $\check{\bfy}$, and take noisy unlabeled samples from $\check{\bfy}$ given by ${\check{\bfx}}=\check{\bfS}\bfB\check{\bfy} +\bfw$, where  $\check{\bfS}$ selects $m$ out of $n$ elements in  $\bfB\check{\bfy}$ randomly and where $\bfw$ is the additive Gaussian noise. We denote the SNR by $\snr=\frac{\|\check{\bfS} \bfB \check{\bfy}\|^2}{\|\bfw\|^2}$ as before. For simulations, we assume an SNR of $20$\,dB.

\subsection{Probability of  Success of the Algorithm}
We run \algname\ with  the noisy input ${\check{\bfx}}$ and a random initialization $\bfS^1 \in \clS$. 
To see the effect of the initialization, we repeat the simulation with a genie-aided initialization of $\bfS^1$ with  $\nu_{\bfS^1,\check{\bfS}}=0.2$, where $20\%$ of the rows of $\bfS^1$ are set equal to the corresponding rows of $\check{\bfS}$ while the remaining rows are selected completely randomly among the remaining possible rows. In both cases, we define the output of \algname\ by ${\bfS^1 \mapsto \bfy^1 \mapsto \bfS^2 \mapsto \bfy^2 \cdots}$ and denote the final output produced by the algorithm by $\bfy^\infty$. We call the recovery successful if the relative error satisfies $\frac{\|\check{\bfy}-\bfy^\infty\|^2}{ \|\check{\bfy}\|^2} \approx O(\frac{1}{\snr})$. 
For simulations, we set $n=1000$ and define parameters $\kapp =\frac{k}{n}$ as the measurement ratio and $\rho=\frac{m}{n}$ as the sampling ratio as before. For each $\kapp$ and $\rho$, we run simulations for $1000$ independent realizations of  $\bfB$ and $\check{\bfS}$ to obtain an estimate of the success probability of \algname. Fig.\,\ref{fig:prob_success} illustrates the success probability as a function of $\kapp, \rho \in (0,1)$ for the fully random initialization in (\ref{fig:rand_init}) and for the genie-aided initialization in (\ref{fig:genie_init}).
The results clearly indicate that a good initialization is crucial for the recovery performance of \algname, as also mentioned in Remark \ref{genie_remark}, where the genie-aided case undergoes a much sharper phase-transition in $\kapp-\rho$ plane.

\subsection{Application to System Identification} 
In this section, as a practical signal processing problem, we study the estimation of the impulse response of a linear time-invariant system, classically known as \textit{System Identification} \cite{ljung1998system}. 
This is illustrated in Fig.\,\ref{fig:deletion_channel}, where a known pre-designed training sequence $\bfb=(b_0, \dots, b_{\tau-1})^\transp$  of length $\tau$ is applied to the input of a linear system with an impulse response of length at most $k$ given by $\bfy=(y_0, \dots, y_{k-1})^\transp$. We assume that an estimate of the delay spread of the channel $k$ is a priori known. The output of the linear system is given by  $\bfz=\bfb\star \bfy$, where $\star$ denotes the convolution operation, where the output $\bfz=(z_0, \dots, z_{n-1})^\transp$ is given by $z_l=\sum_{t=0}^{k-1} y_t b_{l-t}$ for $l=0, 1, \dots, n-1$, where $n=k+\tau-1$ denotes the length of the output $\bfz$  and where $b_r=0$ for $r<0$. 
We consider a scenario in which the output $\bfz$  is observed only through a noisy deletion channel, which deletes some of the output samples $\bfz$ but  preserves their underlying order. Denoting by $\bfy=(y_0, \dots, y_{k-1})^\transp$, we can write $\bfz=\bfB \bfy$ with a measurement matrix $\bfB$ given by

\begin{align}
\bfB=\left ( 
\begin{matrix}
b_0 & 0 & \cdots &0\\
b_1 & b_0 & \ddots & 0\\
b_2 & b_1 & \ddots & \vdots\\
\vdots & \vdots & \ddots & \vdots\\
b_{\tau-1} & b_{\tau-2} & \cdots & b_{\tau-k}\\
0 & b_{\tau-1} & \ddots& \vdots\\
\vdots & \vdots & \ddots & \vdots\\
0 & 0 & \cdots & b_{\tau-1}
\end{matrix}
\right ),\label{B_hankel}
\end{align}
where it is seen that $\bfB$ is an $n\times k$  matrix that depends on the training sequence $\bfb$. We denote the final set of $m$ samples, for some $m\leq n$,  available for system identification by $\bfx=(x_0,\dots, x_{m-1})^\transp$, where $\bfx=\check{\bfS}\bfz+\bfw=\check{\bfS}\bfB\bfy+\bfw$, where $\check{\bfS}$ is a selection matrix representing the location of those samples in $\bfz$ that are not deleted by the deletion channel, and where $\bfw$ is the additive  noise. It is seen that the system identification in the scenario illustrated in Fig.\,\ref{fig:deletion_channel} boils down to the \namesh\ problem \eqref{ul_samp} with a  matrix $\bfB$ given by \eqref{B_hankel}.

\begin{figure}[t]
\centering
\includegraphics[scale=0.9]{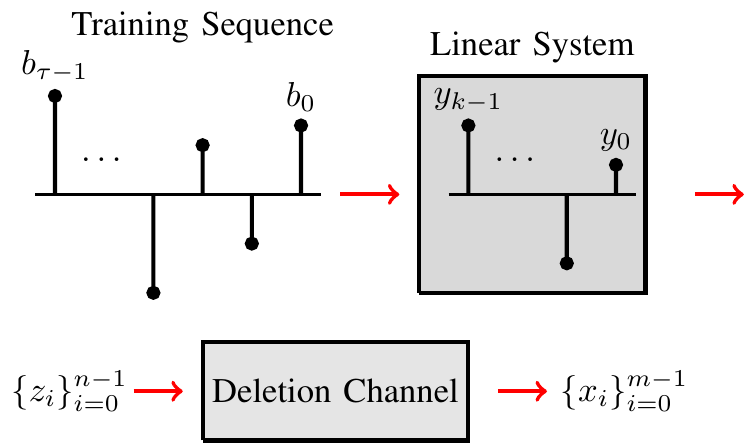}
\caption{Identifying a linear system with an impulse response $\{y_l\}_{l=0}^{k-1}$ of length $k$ via a training sequence $\{b_l\}_{l=0}^{\tau-1}$ of length $\tau$. We assume that some of the output samples $\{z_i\}_{i=0}^{n-1}$ (at unknown positions) are deleted, thus, only a limited number of samples are available for system identification.}
\label{fig:deletion_channel}
\end{figure}

For simulation, we assume that the training sequence $\bfb=(b_0, \dots, b_{\tau-1})^\transp$ has i.i.d. $\sfN(0,1)$ samples and is known to the system identification algorithm (i.e., $\bfB$ in \eqref{B_hankel} is known).  Note that  although the rows of $\bfB$  in \eqref{B_hankel} still consist of Gaussian variables, due to the special structure of the convolution operation, they are highly correlated. Nevertheless, we can still run \algname\ with $\bfB$ as in \eqref{B_hankel}.  Fig.\,\ref{fig:prob_succ_dispersive} illustrates the simulation results for $n=1000$ and for an SNR of $20$\,dB. We  assume that \algname\ is initialized with an $\bfS^1$  with $\nu_{\bfS^1,\check{\bfS}}=0.2$. It is seen that, as expected, for a given delay spread $k$, the performance improves by increasing the length of the training sequence $\tau$ (equivalently $n$)  and the number of unlabeled samples $m$. We also observe that, in comparison with Fig.\,\ref{fig:genie_init} where $\bfB$ has  i.i.d. components across different rows, the correlation among the rows of $\bfB$ degrades the performance of  \algname\ only slightly.

\section{Discussion and further Remarks}\label{sec:discuss}
Fig.\,\ref{fig:prob_success} illustrates the success probability of a single round of \algname. When the success probability is quite small,  \algname\ hits a local minimum with a high probability and fails to find a suitable estimate $\widehat{\bfH}$.
To improve the performance, we can run \algname\ several times each time with a different random initialization $\bfS^1$ and terminate when a \textit{good} estimate is found. However, this requires a procedure to certify whether \algname\ succeeds to find a good estimate.  When the measurement matrix satisfies the R-RIP over $\clH-\clH$, we can develop such a procedure by simply checking whether \algname\ has been able to decode the target signal up to the noise radius, e.g., $\|{\check{\bfx} - \widehat{\bfH}\bbb}\| \leq \eta \|\bfw\|$ for some $\eta=O(1)$ as  in Remark \ref{feas_remark}.

\begin{figure}[t]
	\centering
	\includegraphics[width=0.45\textwidth]{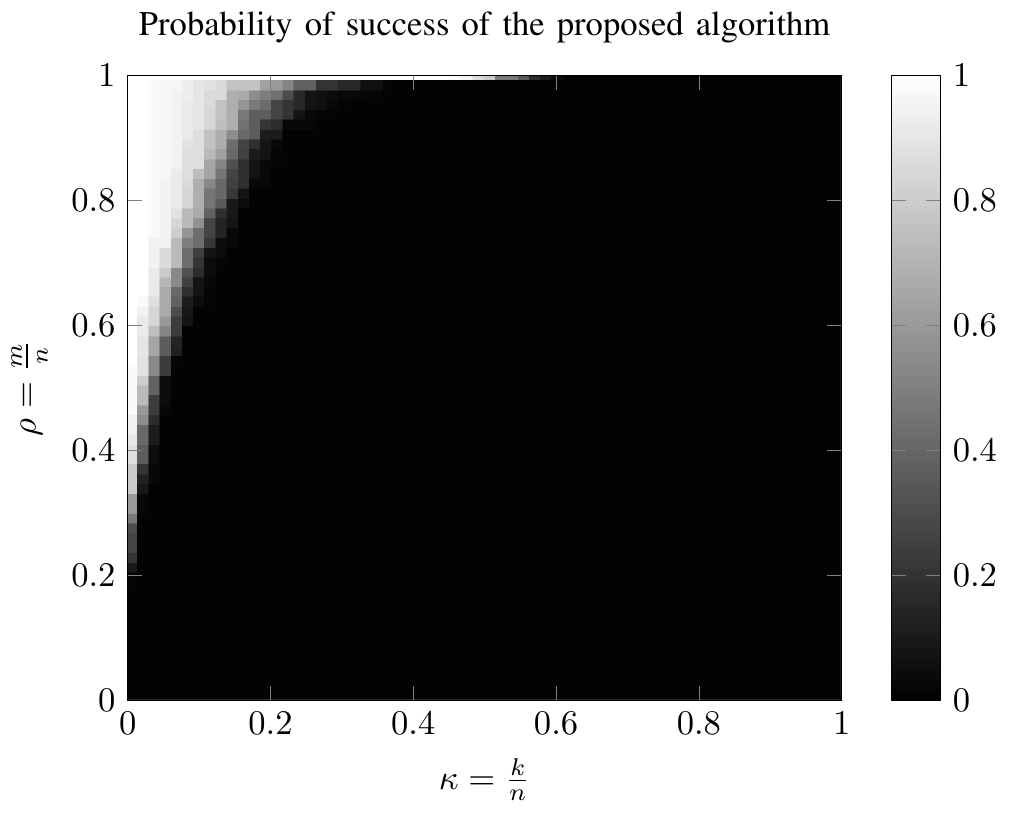}
	\caption{Probability of  success of \algname\ for the estimation of a dispersive channel  as a function of $\kapp=\frac{k}{n}$ and $\rho=\frac{m}{n}$ with a genie-aided  initialization $\bfS^1$ with $\nu_{\bfS^1,\check{\bfS}}=0.2$.}
	\label{fig:prob_succ_dispersive}
\end{figure}

Another direction to improve the performance of \algname\ is to find a \textit{good} initialization for $\bfS^1$. The hope is that such an $\bfS^1$ lies in the basin of attraction of the desired signal $(\check{\bfS},\check{\bfy})$ under \algname. In this paper, we use a random initialization for $\bfS^1$. Assuming that $\bfB$ satisfies the R-RIP over $\clH-\clH$ and ${\nu_{\bfS^1,\check{\bfS}} \geq \nu'>0}$ for a sufficiently large $\nu'$, Theorem \ref{main_thm} together with Theorem \ref{thm_NL} and \ref{thm_NY}   guarantee a suitable recovery of the target signal $(\check{\bfS}, \check{\bfy})$. The condition ${\nu_{\bfS^1,\check{\bfS}} \geq \nu'>0}$ is, however, quite difficult to meet for some $\check{\bfS}$ under the naive random initialization of $\bfS^1$. 
For example, if $\check{\bfS}$ is the selection matrix that samples the first $m$ measurements in $\bfB\check{\bfy}$, then $\bP[\nu_{\bfS^1,\check{\bfS}}\geq \gamma]=\frac{{n-m\gamma \choose m-m\gamma}}{{n \choose m}}$ for a  uniformly randomly sampled $\bfS^1\in \clS$. In the practically interesting regime where ${n,m \to \infty}$ and ${\frac{m}{n} \to \rho \in (0,1)}$, this scales like $e^{-n\varepsilon(\rho, \gamma)}$ with an exponent $\varepsilon(\rho, \gamma)=h(\rho)-(1-\rho\gamma) h(\frac{\rho - \rho \gamma}{1-\rho \gamma})$, where $h(\rho)$ is the entropy function introduced before. A direct calculation reveals that ${\varepsilon(\rho,\gamma)>0}$ for all $\gamma \in (0,1)$, thus, it is almost surely impossible to achieve any  $\gamma\in(0,1)$ asymptotically.
This implies that a \textit{good} initialization method is necessary even when the RIP holds.
Running a nonconvex optimization problem with a good initialization has recently been of interest in other problems in Compressed Sensing such as phase retrieval \cite{candes2015phase}, blind deconvolution \cite{li2016rapid}, and blind calibration \cite{cambareri2016through}. 
We leave developing a good initialization scheme for our algorithm as a future work.

The unlabeled sensing problem $x=\bfS \bfB \bfy +\bfw$ studied in this paper can be extended to cases in which the signal $\bfy$ belongs to a structured class of signals $\clY \subseteq \bR^k$ and $\bfS$ resides in a family of selection matrices $\clS$ other than the ordered sampling matrices studied in this paper.
As in the \namesh, in terms of signal recovery, we need to check two main requirements. 
The first is to develop an R-RIP over $\clH-\clH$ under a suitable metric (e.g., $l_2$ distance as in this paper). 
Following Proposition \ref{rip_2} and assuming that $\bfB$ has i.i.d. Gaussian components, such an R-RIP with a constant ${\ripb \in (0,1)}$ can be generally derived with a probability larger than ${1-2 |\clS|^2 (1+\frac{2}{\ripb}) ^2 e^{-cm\ripb^2\ripc^2}}$, where $|\clS|$ denotes the cardinality of $\clS$.
This provides a theoretical lower bound on the the number of measurements $m$ for a given  $\ripb$.
The second requirement is to develop an algorithm to recover the signal ${\bfH}={\bfy}^\transp \otimes {\bfS} \in \clH$ from the noisy unlabeled measurements ${\bfx}={\bfH} \bbb +\bfw={\bfS} \bfB {\bfy} +\bfw$. From Remark \ref{feas_remark}, under the R-RIP, such an algorithm needs to recover an estimate $\widehat{\bfH} \in \clH$ satisfying  $\|\widehat{\bfH} \bbb - {\bfx}\| \leq \eta \|\bfw\|$ for some $\eta=O(1)$, where $\|.\|$ here denotes the metric with respect to which the RIP is derived. 
The performance guarantee we derived for \namesh\  in Theorem \ref{main_thm} immediately applies to such an estimate.
For the \namesh\ studied in this paper, we used the alternating minimization over $\bfy$ and $\bfS$, where the latter minimization was done with a feasible complexity by using the ordered structure of the matrices in $\clS$ and applying the Dynamic Programming. Deriving such an algorithm for a general signal set $\clH$ (for a general signal  in $\clY$ and unlabeled sampling structure in $\clS$) requires exploiting the algebraic as well as the geometric structure of $\clH$.

\section{Conclusion}
In this paper, we studied the \name\  (\namesh) problem, where the goal was to recover a signal from a set of linear measurements taken via a fully known measurement matrix when the labels of the measurements are missing but their order is  preserved. 
We identified a  duality between \namesh\ and the traditional {Compressed Sensing} (CS), where the unknown support (location of nonzero elements) of a sparse signal in CS corresponds in a natural way to the unknown indices of the measurements kept in \namesh.  Motivated by this duality, we developed a {Restricted Isometry Property} (RIP) similar to that in CS. We also designed a low-complexity Alternating Minimization algorithm to recover the target signal from the set of its noisy unlabeled samples. We analyzed the performance of our proposed algorithm for different signal dimensions and number of measurements theoretically  under the established RIP. We also provided numerical simulations to validate the theoretical results. 

\appendices

\section{Proofs}

\subsection{Proof of Proposition \ref{rip_2}}\label{app:rip_2}

Since R-RIP  is scale-invariant, we first define the following normalized sets
\begin{align*}
\clD_1&:=\big\{(\bfH, \bfH'):  \|\bfH'\|_\sfF^2 \leq \|\bfH\|_\sfF^2=m, d_{\bfH, \bfH'}^2 \geq m \ripc^2\big\},\\
\clD_2&:=\big\{(\bfH, \bfH'):   \|\bfH'\|_\sfF^2 \leq \|\bfH\|_\sfF^2=m\big\},
\end{align*}
where $\clD_1 \subset \clD_2$. 
We also define the following probability events (on the random realization of $\bbb$):
\begin{align*}
\clE_1&:= \bigcup_{\mathclap{(\bfH, \bfH')\in \clD_1}} \big\{\bbb:\big | \|(\bfH-\bfH') \bbb\|^2 - d_{\bfH, \bfH'}^2 \big | > \ripb d_{\bfH, \bfH'}^2 \big\},\\
\clE_2&:=\bigcup_{\mathclap{(\bfH, \bfH')\in \clD_2}} \big \{\bbb: \big | \|(\bfH-\bfH') \bbb\|^2 - d_{\bfH, \bfH'}^2 \big | > \ripb (d_{\bfH, \bfH'}^2\vee m \ripc^2) \big \},
\end{align*}
where for $a,b \in \bR$ we denoted by $a\vee b=\max\{a,b\}$. 
Note that $\clE_1 \subset \clE_2$ since  the subevents corresponding to $(\bfH, \bfH') \in \clD_1$  with $d_{\bfH, \bfH'}^2 \geq m\ripc^2$ (thus, $d_{\bfH, \bfH'}^2 \vee m\ripc^2=d_{\bfH, \bfH'}^2$) in $\clE_1$ are also included in $\clE_2$, where $\clE_2$ contains in addition those subevents corresponding to $d_{\bfH, \bfH'}^2 \leq m \ripc^2$. 
To prove R-RIP result, we need to find an upper bound on $\bP[\clE_1]$. Since $\clE_1 \subset \clE_2$, thus, $\bP[\clE_1] \leq \bP[\clE_2]$, we will do this by deriving an upper bound  for $\bP[\clE_2]$.
So, we focus on the event $\clE_2$ in the sequel. It is seen that $\clE_2$ consists of the union of a continuum of events labeled with $\bfH, \bfH'\in \clD_2$.
As in the proof of Proposition \ref{rip_1}, we will first derive a concentration bound for a fixed $\bfH, \bfH' \in \clD_2$ and then extend it to the whole set $\clD_2$ via a net argument and applying the union bound.

Consider a fixed  $(\bfH, \bfH') \in \clD_2$ where $\bfH=\bfy^\transp\otimes \bfS$ and $\bfH'={\bfy'}^\transp \otimes \bfS'$ with $\|\bfy\|=1$ and $\|\bfy'\|\leq 1$. Let us define $\scrE_{\bfH, \bfH'}=(\bfH-\bfH')\bbb$. Note that $\scrE_{\bfH, \bfH'}$ is an $m$-dim Gaussian vector with a zero mean and a covariance  matrix  
\begin{align}\label{rip_diff_dum1}
\Sigmam&=\bE[\scrE_{\bfH, \bfH'}\scrE_{\bfH, \bfH'}^\transp]=(\bfH-\bfH')(\bfH-\bfH')^\transp\nonumber\\
&\stackrel{(i)}{=}\|\bfy\|^2 \bfI_{m} + \|\bfy'\|^2 \bfI_m -\inp{\bfy}{\bfy'} (\bfS{\bfS'}^\transp + \bfS' \bfS^\transp),
\end{align}
where $\bfI_m$ denotes the identity matrix of order $m$ and where in $(i)$ we used the fact that 
\begin{align}
\bfH\bfH^\transp&=(\bfy^\transp \otimes \bfS)(\bfy^\transp \otimes \bfS)^\transp=(\bfy^\transp \otimes \bfS)(\bfy \otimes \bfS^\transp)\\
&\stackrel{(ii)}{=}(\bfy^\transp \bfy) \otimes (\bfS^\transp \bfS) \stackrel{(iii)}{=}\|\bfy\|^2 \bfS\bfS^\transp \stackrel{(iv)}{=}\|\bfy\|^2 \bfI_m,
\end{align}
where in $(ii)$ we used the property of the Kronecker product, where in $(iii)$ we used the fact that $\bfy^\transp \bfy=\|\bfy\|^2$ is just a number and dropped the Kronecker product, and where in $(iv)$ we used the fact that the each row of $\bfS$ has only one $1$ at a specific column, thus, different rows are orthogonal to each other. A similar derivation gives the second and the third term in \eqref{rip_diff_dum1}. 
It is seen that $\Sigmam$ is not in general a diagonal matrix, thus, $\scrE_{\bfH, \bfH'}$ consists of correlated Gaussian variables and the conventional concentration result for the i.i.d. Gaussian variables does not immediately apply. We first prove that although $\Sigmam$ is not a diagonal matrix, its singular values are bounded and in particular do not grow with the dimension $m$. In words, this implies that the components of $\scrE_{\bfH, \bfH'}$ are not that correlated. 
We denote by $\Sigmam=\bfU \Lambdam \bfU^\transp$ the \textit{Singular Value Decomposition} (SVD) of $\Sigmam$ where $\Lambdam=\diag(\lambda_1, \dots, \lambda_m)$ is the diagonal matrix consisting of the singular values $\lambdam=(\lambda_1, \dots, \lambda_m)^\transp$. We use the convention that the singular values are ordered with $\lambda_m \leq \dots \leq \lambda_1$.
We have the following result.
\begin{lemma}\label{max_eigen_lemma}
	Let $\Sigmam$ and $\lambdam=(\lambda_1, \dots, \lambda_m)^\transp$ be as before. Then,  all the singular values satisfy 
	\begin{align*}
	\|\bfy\|^2 + \|\bfy'\|^2 -2|\inp{\bfy}{\bfy'}| \leq \lambda_i \leq  \|\bfy\|^2 + \|\bfy'\|^2+2|\inp{\bfy}{\bfy'}|.
	\end{align*}
\end{lemma}
\begin{proof}
	Let us denote by ${\sfR:=r_1< \cdots <r_m}$ and $\sfR':=r'_1< \cdots < r'_m$ the ordered sequences consisting of indices of those rows of $\bfB$ selected by $\bfS$ and $\bfS'$, where $r_i, r'_i \in [n]$. 
	Also, let  $\clC:=\{i \in [m]: r_i=r'_i\}$ be the index set of similar elements in $\sfR$ and $\sfR'$.
	Since $\bfS$ and $\bfS'$ have only one $1$ in each row at column set $\sfR$ and $\sfR'$, we can simply check that $\bfS{\bfS'}^\transp$ has at most one $1$ at each row, where $(\bfS{\bfS'}^\transp)_{ij}=1$ if and only if $r_i=r'_j$.
	In particular, the only nonzero diagonal elements of $\bfS{\bfS'}^\transp$ lie on the rows belonging to $\clC$. 
	This implies that the symmetric matrix $\Gammam:=\bfS{\bfS'}^\transp+\bfS'{\bfS}^\transp$ has the diagonal element $2$ and zero off-diagonal terms  at the rows belonging to $\clC$. Moreover, $\Gammam$ has  at most two $1$'s in the other rows (not belonging to $\clC$), where those $1$'s do not lie on the diagonal of $\Gammam$. Therefore, from  \eqref{rip_diff_dum1}, we have the following two cases. On the rows belonging to $\clC$, $\Sigmam$ has only a diagonal element $\|\bfy\|^2 + \|\bfy'\|^2 - 2\inp{\bfy}{\bfy'}$,  which is also a singular values of $\Sigmam$. On the  rows not belonging to $\clC$, $\Sigmam$ has a diagonal element $\|\bfy\|^2 + \|\bfy'\|^2$ plus at most two off-diagonal terms given by $-\inp{\bfy}{\bfy'}$. Hence, form the Gershgorin disk theorem \cite{gershgorin1931uber}, all the singular values of $\Sigmam$ should lie in the range $\|\bfy\|^2+\|\bfy'\|^2 \pm  2|\inp{\bfy}{\bfy'}|$. This completes the proof.
\end{proof}
Let $\Sigmam=\bfU \Lambdam \bfU^\transp$ be the SVD of $\Sigmam$ as before and let $\bfe=\Lambdam^{-\frac{1}{2}} \bfU^\transp \scrE_{\bfH, \bfH'}$. We can check that $\bfe=(e_1, \dots, e_m)^\transp$ consists of i.i.d. $\sfN(0,1)$ variables. Since $\bfU$ is an orthogonal matrix, i.e., $\bfU\bfU^\transp=\bfU^\transp \bfU=\bfI_m$, we have $\|\scrE_{\bfH, \bfH'}\|^2 =\|\bfU^\transp\scrE_{\bfH, \bfH'}\|^2= \sum_{i=1}^m \lambda_i e_i^2$. We also have $\bE[\|\scrE_{\bfH, \bfH'}\|^2]=\|\bfH- \bfH'\|_\sfF^2=d_{\bfH, \bfH'}^2=\sum_{i=1}^m \lambda_i$. Defining $T:=d_{\bfH, \bfH'}^2 + \ripb (d_{\bfH, \bfH'}^2 \vee m \ripc^2)$ and setting $\theta$ to be a positive variable, we obtain the following concentration bound \cite{dubhashi2009concentration}
\begin{align}
\bP\Big[\sum_{i=1}^m& \lambda_i e_i^2 - d_{\bfH, \bfH'}^2> \ripb (d_{\bfH, \bfH'}^2\vee m \ripc^2) \Big ]\nonumber\\
&\leq  e^{- \theta T} \bE[e^{\theta \sum_{i=1}^m \lambda_i  e_i^2}]= \frac{ e^{- \theta T}}{\prod_{i=1}^m \sqrt{1-2\theta \lambda_i}}\nonumber\\
&\leq  e^{- \theta T -\frac{1}{2}\sum_{i=1}^{m} \log(1-2\theta \lambda_i) }=:e^{-E(\theta, \lambdam)},\label{E_thet_def}
\end{align}
where $E(\theta, \lambdam):=\theta T + \frac{1}{2}\sum_{i=1}^{m} \log(1-2\theta \lambda_i)$. We also used the fact that for a $\sfN(0,1)$ variables $e_i$,  $\bE[e^{se_i^2}]=\frac{1}{\sqrt{1-2s}}$ for any $s\in (0,\frac{1}{2})$, thus, the feasible range of $\theta$ in \eqref{E_thet_def} is given by $(0, \frac{1}{2\lambda_{1}})$, where $\lambda_{1}=\max\{\lambda_i: i \in [m]\}$ is the the maximum singular value. From Lemma \ref{max_eigen_lemma}, it results that 
\begin{align}\label{4_upperbound}
\lambda_{1} \leq \|\bfy\|^2 + \|\bfy'\|^2 + 2 |\inp{\bfy}{\bfy'}| \stackrel{(i)}{\leq} 4,
\end{align}
where in $(i)$, we used the fact that $\|\bfy\|=1$ and $\|\bfy'\| \leq 1$ for any $(\bfH, \bfH') \in \clD_2$. This implies that,  for all $(\bfH, \bfH') \in \clD_2$,  the set of permitted values of  $\theta$ at least contains $(0, \frac{1}{8})$. Now let us consider a fixed $\theta \in (0, \frac{1}{8})$. To derive a concentration bound for  \eqref{E_thet_def}, we need to find a strictly positive lower bound on  $E(\theta, \lambdam)$ that is independent of the configuration of the singular values $\lambdam$. We have the following lemma.
\begin{lemma}\label{lam_boundry}
	Let $E(\theta, \lambdam)$ be as before. Then, for any $\theta \in (0, \frac{1}{8})$, the vector $\lambdam$ minimizing $E(\theta, \lambdam)$ (i.e., the worst case singular value configuration) is given by the vector $\lambdam^*:=(4,\dots, 4, 0, \dots, 0)^\transp$ that has $4$ at its $\frac{d_{\bfH, \bfH'}^2}{4}$ components and is $0$ elsewhere.  \hfill $\lozenge$
\end{lemma}
\begin{proof}
	Note that the only constraint we put on $\lambdam$ is that it should belong to the set  
	\begin{align}\label{con_set_lam}
	\Big \{\lambdam: \sum_{i=1}^m \lambda_i = d_{\bfH, \bfH'}^2, 4 \geq \lambda_1 \geq \dots \geq \lambda_m \geq 0 \Big\},
	\end{align}
	where the upper bound $4$ results from \eqref{4_upperbound}.
	Since $T=d_{\bfH, \bfH'}^2 + \ripb(d_{\bfH, \bfH'}^2 \vee m\ripc^2)$ is independent of $\lambdam$,  from $E(\theta, \lambdam)=\theta T + \frac{1}{2}\sum_{i=1}^{m} \log(1-2\theta \lambda_i) $ and the concavity of the Logarithm, it results that $E(\theta, \lambdam)$ is a concave function of $\lambdam$ over the constraint set \eqref{con_set_lam}.
	Therefore, it achieves its minimum at the boundary \cite{bertsekas2003convex} of the constraint set \eqref{con_set_lam}, which corresponds to $\lambdam^*={(4,\dots, 4, 0, \dots, 0)^\transp}$  in the statement of the lemma. This completes the proof.
\end{proof}

From Lemma \ref{lam_boundry}, it results that $E(\theta, \lambdam)$ for $\theta \in (0, \frac{1}{8})$ and for all valid configurations of the singular values  $\lambdam$ is lower bounded by the following function
\begin{align}
E(\theta):=E(\theta, \lambdam^*)=\theta T + \frac{d_{\bfH, \bfH'}^2}{8} \log(1-8 \theta),
\end{align}
where {$\lambdam^*=(4,\dots, 4, 0, \dots, 0)^\transp$} is as in Lemma \ref{lam_boundry}, and where ${T=d_{\bfH, \bfH'}^2 + \ripb(d_{\bfH, \bfH'}^2 \vee m\ripc^2)}$ is as before. The optimal $\theta^* \in (0, \frac{1}{8})$ minimizing $E(\theta)$ is given by $\theta^*=\frac{1}{8} (1-\frac{d_{\bfH, \bfH'}^2}{T})$, where after replacing in $E(\theta)$, yields the following exponent
\begin{align*}
E_{\min}&:=E(\theta^*)=\frac{1}{8}\big(T-{d_{\bfH, \bfH'}^2} - {d_{\bfH, \bfH'}^2} \log(\frac{T}{d_{\bfH, \bfH'}^2})\big)\\
&=\frac{1}{8} \big ( \ripb(d_{\bfH, \bfH'}^2 \vee m\ripc^2) - d_{\bfH, \bfH'}^2 \log(1+ \ripb(1 \vee \frac{m\ripc^2}{d_{\bfH, \bfH'}^2}))\big).
\end{align*}
For $d_{\bfH, \bfH'}^2 \in[ m \ripc^2, \infty)$, $E_{\min}=\frac{1}{8} (\ripb - \log(1+\ripb)) d_{\bfH, \bfH'}^2$, which is an increasing function of $d_{\bfH, \bfH'}^2$ as $\ripb - \log(1+\ripb) \geq 0$ for all $\ripb \in \bR_+$ and in particular $\ripb \in (0,1)$, thus, the minimum of $E_{\min}$ is achieved at $d_{\bfH, \bfH'}^2=m \ripc^2$. 
Similarly, for $d_{\bfH, \bfH'}^2 \in[0, m \ripc^2]$, $E_{\min}=\frac{m\ripc^2 \ripb}{8} (1 - \frac{d_{\bfH, \bfH'}^2}{m\ripc^2 \ripb} \log(1+\frac{m  \ripc^2 \ripb}{d_{\bfH, \bfH'}^2}))$. 
We can check that  $x \mapsto1- x\log(1+\frac{1}{x})$ is a decreasing function of $x$ for $x\in \bR_+$. Thus, setting $x$ equal to $\frac{d_{\bfH, \bfH'}^2}{m  \ripc^2 \ripb}$, it results that the minimum of $E_{\min}$ for {$d_{\bfH, \bfH'}^2 \in [0, m \ripc^2]$} is also achieved at $d_{\bfH, \bfH'}^2= m \ripc^2$. 
Overall, by setting $d_{\bfH, \bfH'}^2= m \ripc^2$, we obtain that $E_{\min}$ is lower bounded by  
\begin{align}
E_0(\ripb, \ripc)&:=\frac{m \ripc^2}{8} (\ripb - \log(1+\ripb)) \nonumber\\
&\stackrel{(i)}{\geq}\frac{m \ripc^2}{8} (\frac{\ripb^2}{2} - \frac{\ripb^3}{3}) \stackrel{(ii)}{\geq} \frac{m\ripc^2 \ripb^2}{48},\label{E_0_bnd}
\end{align}
where in $(i)$ we used the inequality $\log(1+\ripb) \leq \ripb - \frac{\ripb^2}{2} + \frac{\ripb^3}{3}$ for $\ripb \in \bR_+$, and where in $(ii)$ we used $\frac{\ripb^3}{3} \leq \frac{\ripb^2}{3}$ for $\ripb \in (0,1)$. This establishes a strictly positive exponent $E_0(\mu, \ripb)=\frac{m\ripc^2 \ripb^2}{48}$ for \eqref{E_thet_def}, which holds for all $(\bfH, \bfH') \in \clD_2$. 
By following similar steps, we can extend the concentration bound in \eqref{E_thet_def} to the reverse inequality, where overall we obtain that for any fixed $(\bfH, \bfH') \in \clD_2$
\begin{align}
\bP\Big [&\big | \|(\bfH-\bfH') \bbb\|^2 - d_{\bfH, \bfH'}^2 \big | > \ripb (d_{\bfH, \bfH'}^2\vee m \ripc^2) \Big]\nonumber \\
&\leq 2 \bP\Big [ \|(\bfH-\bfH') \bbb\|^2 > d_{\bfH, \bfH'}^2  + \ripb (d_{\bfH, \bfH'}^2\vee m \ripc^2) \Big]\nonumber\\
&\leq 2 e^{-E_0(\ripb, \ripc)}.\label{final_exponent}
\end{align}

The final step is to generalize \eqref{final_exponent}  to derive a  concentration bound for all $(\bfH, \bfH') \in \clD_2$. As in the proof of Proposition \ref{rip_1}, we do this by quantizing $\clD_2$ into a $\frac{\ripb}{2}$-net with a minimal size and applying the union bound. 
We can build such an $\frac{\ripb}{2}$-net by finding a joint net for $\bfy\in \clB_2^k$ and $\bfy' \in \clB_2^k$, each consisting of at most $N=(1+\frac{2}{\ripb})^k$ points (as in the proof of Proposition \ref{rip_1}). Taking the union bound over this joint net and also all possible selection matrices $\bfS$ and $\bfS'$, we obtain that 
\begin{align}
\bP[\clE_2]&\leq |N|^2 |\clS|^2 e^{-E_0(\ripb, \ripc)} \\
&\leq 2 (1+\frac{2}{\ripb})^{2k} {n \choose m}^2 e^{-c m \ripb^2\ripc^2},
\end{align}
where $c>0$ is a constant independent of $m,n,k$ and $\ripb,\ripc$. This completes the proof.

\subsection{Proof of Proposition \ref{two_roots}}\label{app:two_roots}
\noindent We first  need the following two lemmas. We refer to \cite{boyd2004convex} for an introduction to (strict) convexity and (strictly) convex functions needed in this section.
\begin{lemma}\label{lemm_concave}
Let ${\phi(\alpha)=\frac{\cos(\alpha) + \sin(\alpha)-1}{\cos(\alpha) + b}}$, where $b\in \bR_+$ is a constant. Then, $\phi(\alpha)$ is \textit{strictly} concave over $\alpha \in [0,\frac{\pi}{2}]$. \hfill $\square$
\end{lemma}
\begin{proof}
Taking the derivative and simplifying, we obtain
\begin{align}
\phi'(\alpha)=\frac{1+b \cos(\alpha) -(1+b) \sin(\alpha)}{(\cos(\alpha)+b)^2}.
\end{align}
Also, taking the second derivative and simplifying yields
\begin{align}
\phi''(\alpha)&=\frac{b \cos(\alpha) \sin(\alpha) -b(1+b) \cos(\alpha)}{(\cos(\alpha)+b)^3}\\
&+ \frac{(2-b^2)\sin(\alpha) -(1+b) (1+\sin(\alpha)^2)}{(\cos(\alpha)+b)^3}.
\end{align}
Note that for $\alpha \in [0,\frac{\pi}{2}]$ the denominator $(\cos(\alpha)+b)^3$ of both terms is always positive. The numerator of the first term can be written as $b\cos(\alpha)(\sin(\alpha)-1-b)$ which is negative since $\cos(\alpha) \geq 0$, $b\in \bR_+$, and $\sin(\alpha)-1-b \leq 0$. Now let us consider the numerator of the second term. If $\alpha=0$, $\sin(\alpha)=0$ and the numerator is given by $-1-b$, which is negative. So, we consider $\alpha\in (0, \frac{\pi}{2}]$, where $\sin (\alpha) \in (0,1]$. For this case, the numerator of the second term  simplifies to 
\begin{align}
\sin(\alpha) &\Big ( 2-b^2 -(1+b) \big(\frac{1}{\sin(\alpha)} + \sin(\alpha)\big) \Big)\\
&\stackrel{(i)}{\leq} \sin(\alpha) (2-b^2 -(1+b)2)\\
&=\sin(\alpha)(-b^2-2b) \leq 0
\end{align}
where in $(i)$ we used the fact that $b\in \bR_+$, thus, $1+b>0$, and the fact that $\sin(\alpha)> 0$, and applied the inequality $l+ \frac{1}{l} \geq 2$ for  $l\in \bR_+$ by replacing $l=\sin(\alpha)$. This yields  $\phi''(\alpha)\leq 0$ for $a\in [0,\frac{\pi}{2}]$. Moreover, we can check that $\phi''(\alpha)<0$ except at some single points $a$, thus, $\phi(\alpha)$ is \textit{strictly} concave. 
\end{proof}

\begin{lemma}\label{lemm_convex}
Let $\psi(\alpha)=\frac{\cos(\alpha) + b}{\cos(\alpha) + \sin(\alpha)-1}$, where $b\in \bR_+$ is a constant. Then, $\psi(\alpha)$ is convex over $\alpha \in (0,\frac{\pi}{2})$. \hfill $\square$
\end{lemma}
\begin{proof}
First note that $\cos(\alpha) + \sin(\alpha)-1\geq 0$ and $\cos(\alpha)+b\geq 0$ for $\alpha \in (0, \frac{\pi}{2})$, thus, $\psi(\alpha)\geq0$. We can write ${-\psi(\alpha)=\iota \circ \phi(\alpha)}$ where ${\iota(x)=-\frac{1}{x}}$ and  is a concave increasing function over $(0, \infty)$, which contains $\phi(\alpha)$ for $\alpha\in (0, \frac{\pi}{2})$ (since $\phi(\alpha)>0$ in this domain). Let $\alpha,\alpha' \in (0,\frac{\pi}{2})$ and  $\lambda \in [0,1]$, and set $\bar{\lambda}=1-\lambda$. We have
\begin{align}
-\psi(\lambda \alpha+ \bar{\lambda} \alpha')&=\iota \circ \phi(\lambda \alpha+ \bar{\lambda} \alpha')\\
& \stackrel{(i)}{\geq} \iota (\lambda \phi(\alpha) + \bar{\lambda} \phi(\alpha'))\\
& \stackrel{(ii)}{\geq} \lambda \iota \circ \phi(\alpha) + \bar{\lambda} \iota \circ \phi(\alpha')\\
& = \lambda (-\psi(\alpha))  + \bar{\lambda} (-\psi(\alpha')),
\end{align}
where in $(i)$ we used the concavity of $\phi$ proved in Lemma \ref{lemm_concave} and the fact that $\iota$ is an increasing function, and where in $(ii)$ we used the concavity of $\iota$. This implies that $-\psi(\alpha)$ is a concave function, thus, $\psi(\alpha)$ is  convex over $\alpha\in(0,\frac{\pi}{2})$. 
\end{proof}

We can now prove Proposition \ref{two_roots}. First note that from \eqref{upsi_fx}
\begin{align*}
\upsilon(\alpha)=\frac{(\vsigma-1)\cos(\alpha)+ \vrho}{\cos(\alpha)+\sin(\alpha)-1}-\frac{\cos(\alpha)+\sin(\alpha)-1}{(\vsigma+1)\cos(\alpha)+ \vrho}. 
\end{align*}
As ${\vsigma -1 > 0}$ and $\vrho > 0$, the first term is a strictly convex function of $\alpha\in (0, \frac{\pi}{2})$ from Lemma \ref{lemm_convex} (by setting $b=\frac{\vrho}{\vsigma-1} > 0$). The second term is also a concave function of $\alpha$ from Lemma \ref{lemm_concave} (by setting $b=\frac{\vrho}{\vsigma+1} > 0$). Therefore, $\upsilon(\alpha)$ is strictly convex in $(0, \frac{\pi}{2})$. 

As ${\upsilon(0^+)=\upsilon (\frac{\pi}{2}^-)=+\infty}$, the roots of $\upsilon$ should lie in $(0, \frac{\pi}{2})$. Suppose that $\upsilon(\alpha)$ has more than two distinct roots, i.e., there are $\alpha_1 < \alpha_2 < \alpha_3$ in $(0, \frac{\pi}{2})$ with $\upsilon(\alpha_i)=0$, for $i=1,2,3$. Then, setting $\lambda=\frac{\alpha_2-\alpha_1}{\alpha_3-\alpha_1} \in (0,1)$, a simple calculation shows that $\alpha_2=\lambda \alpha_3 + (1-\lambda) \alpha_1$. Thus, we have
\begin{align}
0=\upsilon(\alpha_2)&=\upsilon(\lambda \alpha_3 + (1-\lambda) \alpha_1) \nonumber \\
&\stackrel{(i)}{<} \lambda \upsilon(\alpha_3) + (1-\lambda) \upsilon(\alpha_1)=0,\label{contradict}
\end{align}
where in $(i)$ we used the the fact that $\lambda \in (0,1)$ and that $\upsilon(\alpha)$ is strictly convex for $\alpha \in(0,1)$. Since \eqref{contradict} is a contradiction, $\upsilon(\alpha)$ cannot have more than two roots.

\subsection{Proof of Proposition \ref{root_prop}}\label{app:root_prop}
First note that $\upsilon(\alpha, \vsigma, \vrho)$ is a continuous and differentiable function of $(\alpha, \vsigma, \vrho)\in (0, \frac{\pi}{2}) \times (1, \infty)\times (0,1)$ of any order. 

To prove $(i)$, let $(\vsigma_0, \vrho_0)\in \Dxirho$ be an arbitrary point and let $\alpha_{\min}(\vsigma_0, \vrho_0)< \alpha_{\max}(\vsigma_0, \vrho_0)$ be the corresponding two roots. From the strict convexity of $\upsilon$ proved in Proposition \ref{two_roots}, we have $\upsilon(\alpha_0, \vsigma_0, \vrho_0) <0$ for $\alpha_0=\frac{\alpha_{\min}(\vsigma_0, \vrho_0) +\alpha_{\max}(\vsigma_0, \vrho_0)}{2}$. Thus, from the continuity of $\upsilon(\alpha_0, \vsigma, \vrho)$ it results that there is an open set $\clD_0:=(\vsigma_0-\varepsilon, \vsigma_0+\varepsilon) \times (\vrho_0 -\varepsilon, \vrho_0+\varepsilon)$, with a sufficiently small $\varepsilon>0$,  around $(\vsigma_0, \vrho_0)$ over which $\upsilon(\alpha_0, \vsigma, \vrho) <0$. For all  $(\vsigma, \vrho) \in \clD_0$, $\upsilon(\alpha, \vsigma, \vrho)$ must have two roots from Proposition \ref{two_roots}. This implies that ${\clD_0 \subset \Dxirho}$, thus, $\Dxirho$ is an open set.

To prove $(ii)$, note that $\upsilon(\alpha, \vsigma, \vrho)=0$ defines $\alpha_{\min}$ and $\alpha_{\max}$ implicitly as  functions of $(\vsigma, \vrho) \in \Dxirho$. Since $\Dxirho$ is an open set from $(i)$,  Implicit Function Theorem \cite{dontchev2009implicit} implies that $\alpha_{{\max/\min}}$ would be differentiable at a point $(\vsigma_0, \vrho_0) \in \Dxirho$ provided that $\frac{\partial}{\partial \alpha} \upsilon(\alpha, \vsigma_0, \vrho_0)=\upsilon'(\alpha, \vsigma_0, \vrho_0) \neq 0$  at ${\alpha=\alpha_{{\max/\min}} (\vsigma_0, \vrho_0)}$. So, we need to prove that this condition is satisfied (see, e.g., Fig.\,\ref{fig:upsi_fx} where $\upsilon'(\alpha, \vsigma_0, \vrho_0)$ is strictly positive at $\alpha_{\max}$ and strictly negative at $\alpha_{\min}$). Suppose, for example, that $\upsilon'(\alpha_{\min}, \vsigma_0, \vrho_0)=0$. Then, from the strict convexity of $\upsilon(\alpha, \vsigma_0, \vrho_0)$ proved in Proposition \ref{two_roots}, it results that $\upsilon'(\alpha, \vsigma_0, \vrho_0)>0$ for $\alpha \in (\alpha_{\min}, \alpha_{\max})$, which from $\upsilon(\alpha_{\min}, \vsigma_0, \vrho_0)=0$ implies that $\upsilon(\alpha_{\max}, \vsigma_0, \vrho_0)>0$, thus, contradicting the fact that $\alpha_{\max}$ is another root of $\upsilon(\alpha, \vsigma_0, \vrho_0)$. A similar argument shows that $\upsilon'(\alpha_{\max}, \vsigma_0, \vrho_0)\neq 0$. Therefore, both $\alpha_{\min}$ and $\alpha_{\max}$ are differentiable  in $\Dxirho$.

To prove $(iii)$, we  first check that $\upsilon(\vsigma, \vrho)$ is an increasing function of $\vsigma$ and $\vrho$. This  follows simply from the fact that $\cos(\alpha)\geq 0$ for $\alpha \in [0, \frac{\pi}{2}]$, thus, the numerator/denominator of the first/second term in \eqref{upsi_fx} is an increasing function of $(\vsigma, \vrho)$. This implies that $\frac{\partial}{\partial \vsigma} \upsilon(\alpha, \vsigma, \vrho)>0$ and $\frac{\partial}{\partial \vrho} \upsilon(\alpha, \vsigma, \vrho)>0$ for all $\alpha \in (0, \frac{\pi}{2})$ including $\alpha_{\min/\max}$. Using the Implicit Function Theorem \cite{dontchev2009implicit} and taking the derivative with respect to $\vsigma$ from $\upsilon(\alpha_{\min}(\vsigma, \vrho), \vsigma, \vrho)=0$ yields  
\begin{align*}
\upsilon'(\alpha_{\min}, \vsigma, \vrho) \frac{\partial}{\partial \vsigma} \alpha_{\min}(\vsigma, \vrho) =-\frac{\partial}{\partial \vsigma} \upsilon(\alpha, \vsigma, \vrho)\big |_{\alpha=\alpha_{\min}}\leq 0.
\end{align*}
From $\upsilon'(\alpha_{\min}, \vsigma, \vrho)<0$ established in $(ii)$ (see also Fig.\,\ref{fig:upsi_fx}), this implies that $\frac{\partial}{\partial \vsigma} \alpha_{\min}(\vsigma, \vrho)>0$. Similarly, we obtain that $\frac{\partial}{\partial \vrho} \alpha_{\min}(\vsigma, \vrho)>0$. Thus, $\alpha_{\min}(\vsigma, \vrho)$ is an increasing function of $(\vsigma, \vrho)$. The result for $\alpha_{\max}(\vsigma, \vrho)$ follows similarly with the only difference that $\upsilon'(\alpha_{\max}, \vsigma, \vrho)>0$ (see  Fig.\,\ref{fig:upsi_fx}), which implies that $\alpha_{\max}(\vsigma, \vrho)$ is a decreasing function of $(\vsigma, \vrho)$.

To prove $(iv)$, it is easier to use \eqref{eq:fx_pt}. Note that 
\begin{align}
(\cos(\alpha)&+\sin(\alpha)-1)^2=(\vsigma \cos(\alpha)+\vrho )^2-\cos(\alpha)^2\nonumber\\
&=(\vsigma^2 -1) \cos(\alpha)^2 + 2 \vrho \vsigma \cos(\alpha) + \vrho^2\nonumber\\
&\leq (\vsigma^2 -1) + 2 \vrho \vsigma+\vrho^2 =(\vsigma+\vrho)^2-1,\label{dm_rts}
\end{align}
where the last term converges to $0$ as ${(\vsigma, \vrho)\to (1,0)}$. Let us define the open region $\clD=\{(\vsigma, \vrho): (\vsigma+\vrho)^2-1< \varepsilon^2\}$ with $\varepsilon \in (0, \sqrt{2}-1)$. For any $(\vsigma, \vrho) \in \clD$, it is seen from \eqref{dm_rts} that \eqref{eq:fx_pt} has a solution satisfying $\cos(\alpha)+ \sin(\alpha) < 1+\varepsilon$. Note that such a solution exists  since $1+\varepsilon \in (1, \sqrt{2})$, which is a subset of the range of $\cos(\alpha)+\sin(\alpha)-1$ for $\alpha \in [0, \frac{\pi}{2}]$ given by $[1, \sqrt{2}]$. Thus, $\clD\subseteq \Dxirho$. In particular, $\clD$ contains all the paths along which $(\vsigma, \vrho)$ can approach to $(1,0)$ from the inside of $\Dxirho$. From the identity $\cos(\alpha)+\sin(\alpha)=\sqrt{2} \cos(\alpha-\frac{\pi}{4})$, it results that for any $(\vsigma, \vrho) \in \clD$
\begin{align}
\alpha_{\max}(\vsigma, \vrho)& \geq \frac{\pi}{4} + \cos^{-1}(\frac{1+\varepsilon}{\sqrt{2}}),\\
\alpha_{\min}(\vsigma, \vrho)& \leq  \frac{\pi}{4} - \cos^{-1}(\frac{1+\varepsilon}{\sqrt{2}}),
\end{align}
Taking the limit as $\varepsilon \to 0$ yields that any limit point of $\alpha_{\max}(\vsigma, \vrho)$ as ${(\vsigma, \vrho)\to (1,0)}$ should be larger than $\lim_{\varepsilon \to 0} \frac{\pi}{4} + \cos^{-1}(\frac{1+\varepsilon}{\sqrt{2}})=\frac{\pi}{2}$. Since ${\alpha \in [0,\frac{\pi}{2}]}$, this implies that $\lim_{{(\vsigma, \vrho)\to (1,0)}}\alpha_{\max} (\vsigma, \vrho)=\frac{\pi}{2}$. Similarly, it results that $\lim_{{(\vsigma, \vrho)\to (1,0)}}\alpha_{\min} (\vsigma, \vrho)=0$. This completes the proof.

\balance
{\small
\bibliographystyle{IEEEtran}
\bibliography{references}
}

\end{document}